\documentclass[onefignum,onetabnum,final]{siamart171218}

\usepackage{verbatim}
\usepackage{microtype}
\usepackage{tikz}
\usepackage[colorinlistoftodos,prependcaption]{todonotes}
\usepackage{amsmath}
\usepackage{amssymb}
\usepackage[absolute]{textpos}

\usepackage[
  margin=3.7cm,
  includefoot,
  footskip=30pt,
]{geometry}
\usepackage{layout}

\newcommand{\mc}{\mathcal}
\newcommand{\Oh}{\mc{O}}

\newcommand{\Cc}{\mc{C}}
\newcommand{\Gg}{\mc{G}}
\newcommand{\Qq}{\mc{Q}}
\newcommand{\YYY}{\mc{Y}}
\newcommand{\FFF}{\mc{F}}

\newcommand{\N}{\mathbb{N}}
\newcommand{\R}{\mathbb{R}}

\newcommand{\projprof}{\widehat{\mu}}
\newcommand{\Oof}{\mc{O}}

\newcommand{\fwcol}{f_{\wcol}}

\newcommand{\fproj}{f_{\mathrm{proj}}}

\newcommand{\Opt}{\mathrm{OPT}}
\newcommand{\Cds}{\mathrm{CDS}}

\newcommand{\wcol}{\mathrm{wcol}}

\newcommand{\WReach}{\mathrm{WReach}}

\newtheorem*{rrule*}{Reduction Rule A1}
\newtheorem*{rrrule*}{Reduction Rule B1}
\newtheorem{claim}{Claim}[theorem]

\def\grad_#1{\nabla\!_{#1}}
\def\topgrad_#1{\widetilde \nabla\!_#1}

\newcommand{\ds}{\mathbf{ds}}
\newcommand{\cds}{\mathbf{cds}}
\newcommand{\stt}{\mathbf{st}}
\newcommand{\cl}{\text{cl}}
\newcommand{\pr}[3]{M_{#1}(#2,#3)}
\newcommand{\prg}[4]{M^{#4}_{#1}(#2,#3)}

\title{Lossy Kernels for Connected Dominating Set on
Sparse Graphs\thanks{
The work of Sebastian Siebertz is supported by the National Science Centre of
Poland via POLONEZ grant agreement UMO-2015/19/P/ST6/03998, which has received funding from the European Union's Horizon 2020 research and innovation programme (Marie Sk\l odowska-Curie grant agreement No. 665778). Eduard Eiben was supported by Pareto-Optimal Parameterized Algorithms (ERC Starting Grant 715744) and by the Austrian Science Fund (FWF, projects P26696 and W1255-N23).
}}

\headers{Lossy Kernels for Connected Dominating Set on
Sparse Graphs}{E. Eiben, M. Kumar, A.E. Mouawad, F. Panolan, S. Siebertz} 

\author{Eduard Eiben$^\ddagger$\thanks{Algorithms and Complexity Group, TU Wien, Austria. 
\texttt{eiben@ac.tuwien.ac.at}}\and 
Mithilesh Kumar\thanks{Department of Informatics, University of Bergen, Norway.
\newline
	\texttt{\{mithilesh.kumar,a.mouawad,fahad.panolan\}@ii.uib.no}}\and
Amer E. Mouawad\footnotemark[3]
\and Fahad Panolan\footnotemark[3] \and
Sebastian Siebertz\thanks{Faculty of Mathematics, Informatics and Mechanics, University of Warsaw, Poland. 
\texttt{siebertz@mimuw.edu.pl}}}

\usepackage{amsopn}

\begin{document}
\maketitle

\begin{textblock}{5}(11.65, 12.3)
\includegraphics[width=45px]{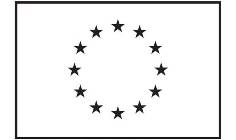}%
\end{textblock}

\begin{abstract}
For $\alpha > 1$, an $\alpha$-approximate (bi-)kernel is a 
polynomial-time algorithm that takes as input an instance $(I, k)$ of
a problem $\Qq$ and outputs an instance $(I',k')$ (of a problem $\Qq'$) 
of size bounded by a function of $k$ such that, for every $c\geq 1$, a
$c$-approximate solution for the new instance can be turned into a
$(c\cdot\alpha)$-approximate solution of the original instance in
polynomial time. This framework of \emph{lossy kernelization} was
recently introduced by Lokshtanov et al. 
We study {\sc Connected Dominating Set} (and its distance-$r$ variant)
parameterized by solution size on sparse graph classes like
biclique-free graphs, classes of bounded expansion, and nowhere dense
classes. We prove that for every $\alpha>1$, {\sc Connected Dominating
Set} admits a polynomial-size $\alpha$-approximate (bi-)kernel on all
the aforementioned classes.
Our results are in sharp contrast to the kernelization complexity of
{\sc Connected Dominating Set}, which is known to not admit a
polynomial kernel even on $2$-degenerate graphs and graphs of bounded
expansion, unless $\textsf{NP} \subseteq \textsf{coNP/poly}$.  We
complement our results by the following conditional lower bound.  We
show that if a class~$\Cc$ is somewhere dense and closed under taking
subgraphs, then for some value of $r\in\N$ there cannot exist an
$\alpha$-approximate bi-kernel for the {\sc (Connected) Distance-$r$
Dominating Set} problem on $\Cc$ for any $\alpha>1$ (assuming the Gap
Exponential Time Hypothesis).
\end{abstract}

\section{Introduction}\label{sec:intro}


A powerful method in parameterized complexity theory is to compute on
input $(I,k)$ a problem \emph{kernel} in a polynomial-time 
pre-processing step, that is, to reduce the input instance in
polynomial time to an equivalent instance $(I',k')$ of size $g(k)$ for
some function $g$ bounded in the parameter only. If the reduced
instance $(I',k')$ belongs to a different problem than $(I,k)$, we
speak of a \emph{bi-kernel}. It is well known that a problem is
fixed-parameter tractable if and only if it admits a kernel, however,
in general the function~$g$ can grow arbitrarily fast. For practical
applications we are mainly interested in linear or at worst polynomial
kernels. We refer to the
textbooks~\cite{cygan2015parameterized,downey2013fundamentals,
  downey1999parameterized} for extensive background on parameterized
complexity and kernelization.

One shortcoming of the above notion of kernelization is that it does
not combine well with approximation algorithms or heuristics. An
approximate solution on the reduced instance provides no insight
whatsoever about the original instance, the only statement we can
derive from the definition of a kernel is that the reduced instance
$(I',k')$ is a positive instance if and only if the original instance
$(I, k)$ is a positive instance. This issue was recently addressed by
Lokstanov et al.~\cite{LokshtanovPRS16}, who introduced the framework
of \emph{lossy kernelization}. Intuitively, the framework combines
notions from approximation and kernelization algorithms to allow for
approximation preserving kernels.

\pagebreak
Formally, a \emph{parameterized optimization (minimization or
  maximization) problem}~$\Pi$ over finite vocabulary $\Sigma$ is a
computable function
$\Pi\colon\Sigma^\star\times \N \times \Sigma^\star\rightarrow
\R\cup\{\pm \infty\}$.
A \emph{solution} for an instance $(I,k)\in \Sigma^\star\times\N$ is a
string $s\in \Sigma^\star$, such that $|s| \leq |I| + k$. The
\emph{value} of the solution $s$ is $\Pi(I,k,s)$. For a minimization
problem, the \emph{optimum value} of an instance $(I,k)$ is
$\Opt_\Pi(I,k)=\min_{s\in \Sigma^*, |s|\leq |I|+k}\Pi(I,k,s)$, for a
maximization problem it is
$\Opt_\Pi(I,k)=\max_{s\in \Sigma^*, |s|\leq |I|+k}\Pi(I,k,s)$. An
\emph{optimal solution} is a solution $s$ with
$\Pi(I,k,s)=\Opt_\Pi(I,k)$. If $\Pi$ is clear from the context, we
simply write $\Opt(I,k)$.

A vertex-subset graph problem~$\mathcal{Q}$ defines which subsets of 
the vertices of an input graph are feasible solutions. We consider the following
parameterized minimization problem associated with $\mathcal{Q}$:
\[\mathcal{Q}(G,k,S)=
\begin{cases}
  \infty & \text{if $S$ is not a valid solution for $G$}\\
  & \hspace{0.5cm}\text{as determined
    by $\mathcal{Q}$}\\\min\{|S|,k+1\} & \text{otherwise.}
\end{cases}\]
Note that this bounding of the objective function at $k+1$ does not 
make sense for approximation algorithms if one insists on $k$ being the 
unknown optimum solution of the input instance.
Also, notice that in standard kernelization or fixed parameter tractable algorithms, when we parameterize by solution size $k$, 
for a decision version of a minimization problem, we actually do not care about solutions of size more than $k$. 
However, we always aim for efficient algorithms, where efficiency is measured in terms of $k$. 
Going by the same logic, we set ${\mathcal Q}(G,k,S)=k+1$, when $|S| \geq k+1$, so that all solutions 
of cardinality more than $k$ are ``equally bad'' or indistinguishable. 
The symbol $\infty$ is used to distinguish between actual solutions and other strings which are not solutions.  
For detailed discussion about capping the objective function at $k+1$, we refer to~\cite{LokshtanovPRS16}.    

\begin{definition}
  Let $\alpha>1$ and let $\Pi$ be a parameterized minimization
  problem. An \emph{$\alpha$-approximate polynomial time
    pre-processing algorithm} $\mathcal{A}$ for $\Pi$ is a pair of
  polynomial time algorithms.  The first algorithm is called the
  \emph{reduction algorithm}, and computes a map
  $R_\mathcal{A} \colon \Sigma^\star\times \N\rightarrow
  \Sigma^\star\times \N$.
  Given as input an instance $(I, k)$ of $\Pi$, the reduction
  algorithm outputs another instance $(I',k')=R_\mathcal{A}(I,k)$.
  The second algorithm is called the \emph{solution lifting
    algorithm}. It takes as input an instance
  $(I,k)\in \Sigma^\star\times \N$, the output instance $(I',k')$ of
  the reduction algorithm, and a solution $s'$ to the instance
  $(I',k')$.  The solution lifting algorithm works in time polynomial
  in $|I|,k, |I'|, k'$ and $s'$, and outputs a solution $s$ to
  $(I, k)$ such that
\begin{align*}
  \frac{\Pi(I,k,s)}{\Opt(I,k)}\leq \alpha\cdot \frac{\Pi(I',k',s')}{\Opt(I',k')}. 
\end{align*}
\end{definition}

\begin{definition}
  An \emph{$\alpha$-approximate kernelization algorithm} is an
  $\alpha$-approximate polynomial time pre-processing algorithm for
  which we can prove an upper bound on the size of the output
  instances in terms of the parameter of the instance to be
  pre-processed. We speak of a linear or polynomial kernel, if the
  size bound is linear or polynomial, respectively. If we allow the
  reduced instance to be an instance of another problem, we speak of
  an \emph{$\alpha$-approximate bi-kernel}.  A \emph{polynomial-size
    approximate kernelization scheme (PSAKS)} is a family of
  $\alpha$-approximate polynomial kernelization algorithms for each
  $\alpha > 1$, with size of each kernel bounded by
  $f(\alpha)k^{g(\alpha)}$, for some functions $f$ and $g$ independent
  of $|I|$ and $k$.
\end{definition}

We refer to the work of Lokshtanov et al.~\cite{LokshtanovPRS16} for
an extensive discussion of related work and examples of problems that
admit lossy kernels.

\subparagraph*{Sparse graphs and domination}

We consider finite, undirected and simple graphs and refer to the
textbook~\cite{diestel-book} for all undefined notation. We write
$K_{i,j}$ for the complete bipartite graph with partitions of size $i$
and $j$, respectively. We call a class $\Cc$ of graphs biclique-free
if there are $i,j\in\N$ such that $K_{i,j}$ is not a subgraph of $G$
for all $G\in\Cc$.

The notion of nowhere denseness was introduced by Ne\v set\v ril and
Ossona de Mendez~\cite{nevsetvril2010first,nevsetvril2011nowhere} as a
general model of \emph{uniform sparseness} of graphs. Many familiar
classes of sparse graphs, like planar graphs, graphs of bounded
tree-width, graphs of bounded degree, and all classes that exclude a
fixed (topological) minor, are nowhere dense. An important and related
concept is the notion of a graph class of \emph{bounded expansion},
which was also introduced by Ne\v set\v ril and Ossona de
Mendez~\cite{nevsetvril2008grad,nevsetvril2008gradb,nevsetvril2008gradc}.

\begin{definition}
  Let $H$ be a graph and let $r\in \N$. An \emph{$r$-subdivision} of
  $H$ is obtained by replacing all edges of $H$ by internally vertex
  disjoint paths of length at most $r$.
\end{definition}

\begin{definition}
  A class $\Cc$ of graphs is \emph{nowhere dense} if there exists a
  function $t\colon \N\rightarrow \N$ such that for all $q\in\N$ and
  for all $G\in \Cc$ we do not find an $q$-subdivision of the
  complete graph~$K_{t(q)}$ as a subgraph of $G$. Otherwise, $\Cc$ is
  called \emph{somewhere dense}.
\end{definition}

\begin{definition}
  A class $\Cc$ of graphs has \emph{bounded expansion} if there exists
  a function $d\colon \N\rightarrow \N$ such that for all $r\in\N$ and
  all graphs $H$, such that an $r$-subdivision of $H$ is a subgraph of
  $G$ for some $G\in\Cc$, satisfy $|E(H)|/|V(H)|\leq d(r)$.
\end{definition}

Every class of bounded expansion is nowhere dense, which in turn
excludes some biclique as a subgraph and hence is biclique-free.  For
extensive background on bounded expansion and nowhere dense graphs we
refer to the textbook of Ne\v{s}et\v{r}il and Ossona de
Mendez~\cite{sparsity}.

\begin{definition}
  In the parameterized \textsc{Dominating Set} (DS) problem we are
  given a graph~$G$ and an integer parameter $k$, and the objective is
  to determine the existence of a subset $D\subseteq V(G)$ of size at
  most $k$ such that every vertex $u$ of $G$ is \emph{dominated} by
  $D$, that is, either~$u$ belongs to~$D$ or has a neighbor in~$D$.
  More generally, for fixed $r\in \N$, in the \textsc{Distance-$r$
    Dominating Set} ($r$-DS) problem we are asked to determine the
  existence of a subset~$D\subseteq V(G)$ of size at most $k$ such
  that every vertex $u\in V(G)$ is within distance at most~$r$ from a
  vertex of~$D$. In the \textsc{Connected (Distance-$r$) Dominating
    Set} (CDS/$r$-CDS) problem we additionally demand that the
  (distance-$r$) dominating set shall be connected.
\end{definition}

The {\sc Dominating Set} problem plays a central role in the theory of
parameterized complexity, as it is a prime example of a
$\mathsf{W}[2]$-complete problem with the size of the optimal solution
as the parameter, thus considered intractable in full generality.  For
this reason, the {\sc (Connected) Dominating Set} problem and
{\sc (Connected) Distance-$r$ Dominating Set} problem have been extensively
studied on restricted graph classes.  A particularly fruitful line of
research in this area concerns kernelization algorithms for the aforementioned 
problems~\cite{AlberFN04,BodlaenderFLPST09,FominLST10,FominLST12,FominLST13,PhilipRS12}.

Philip et al.~\cite{PhilipRS12} obtained a kernel of size
$\Oh(k^{(d+1)^2})$ on $d$-degenerate graphs, for constant~$d$, and
more generally a kernel of size $\Oh(k^{\max(i^2,j^2)})$ on graphs
excluding the biclique $K_{i,j}$ as a subgraph. On the lower bounds
side, Cygan et al.~\cite{CyganGH13} have shown that the existence of a
size $\Oh(k^{(d-1)(d-3)-\epsilon})$ kernel, $\epsilon> 0$, for {\sc
  Dominating Set} on $d$-degenerate graphs would imply \textsf{NP}
$\subseteq$ \textsf{coNP/poly}.  For the \textsc{Connected Dominating
  Set} problem linear kernels are only known for
planar~\cite{LokshtanovMS11} and $H$-topological-minor-free
graphs~\cite{FominLST13}. Polynomial kernels are excluded already for
graphs of bounded degeneracy~\cite{CyganPPW12}, assuming \textsf{NP}
$\not\subseteq$ \textsf{coNP/poly}.

\pagebreak
For the more general {\sc Distance-$r$ Dominating Set} problem we know the 
following results. Dawar and Kreutzer~\cite{DawarK09} showed that for
every $r\in \N$ and every nowhere dense class~$\Cc$, the {\sc Distance-$r$ Dominating Set} problem is fixed-parameter tractable on $\Cc$. 
Drange et al.~\cite{DrangeDFKLPPRVS16} gave a linear bi-kernel for {\sc Distance-$r$ Dominating Set} on any graph class of bounded expansion
for every $r \in \N$, and a pseudo-linear kernel for {\sc Dominating Set} on 
any nowhere dense graph class; that is, a kernel of size
$f(\epsilon)\cdot k^{1+\epsilon}$, for some function $f$.  
Precisely, the kernelization algorithm of
Drange et al.~\cite{DrangeDFKLPPRVS16} outputs an instance of an
annotated problem where some vertices are not required to be
dominated; this will be the case in the present paper as
well (except for the case of biclique-free graphs). Kreutzer et al.~\cite{siebertz2016polynomial} provided a
polynomial bi-kernel for the {\sc Distance-$r$ Dominating Set} problem on
every nowhere dense class for every fixed $r\in \N$ and finally,
Eickmeyer et al.~\cite{eickmeyer2016neighborhood} could prove the
existence of pseudo-linear bi-kernels of size $f(r,\epsilon)\cdot k^{1+\epsilon}$, for some function $f$ and any $\epsilon > 0$. 

It is known that bounded expansion classes of graphs are the limit for
the existence of polynomial kernels for the {\sc Connected Dominating Set} 
problem. Drange et al.~\cite{DrangeDFKLPPRVS16} gave an example of a
subgraph-closed class of bounded expansion which does not admit a
polynomial kernel for {\sc Connected Dominating Set} unless
$\mathsf{NP}\subseteq \mathsf{coNP/Poly}$. They also showed that
nowhere dense classes are the limit for the fixed-parameter
tractability of the {\sc Distance-$r$ Dominating Set} problem if we assume
closure under taking subgraphs (classes which are closed under taking subgraphs will be called \emph{monotone classes}).

\medskip
\paragraph{Our results}
In this paper we prove the following results. 

\medskip
\begin{itemize}
\item For every $\alpha>1$, {\sc Connected Dominating Set} admits an
  $\alpha$-approximate kernel on $K_{d,d}$-free graphs of size
  $k^{\Oh(\frac{d^2}{\alpha-1})}$.
\item For every $\alpha>1$, {\sc Connected Dominating Set} admits an
  $\alpha$-approximate bi-kernel on graphs of bounded expansion of
  size $\Oh(f(\alpha) \cdot k)$ (i.e, linear in $k$), where
  $f(\alpha)$ is a function depending only on $\alpha$.
\item For every $\alpha>1$ and every $r\in\N$, {\sc Connected
    Distance-$r$ Dominating Set} admits an $\alpha$-approximate
  bi-kernel on classes of nowhere dense graphs of size polynomial in
  the parameter, where the degree of the polynomial depends only on
  $r$ and not on $\alpha$.
\end{itemize}

\smallskip
Observe that our results are in sharp contrast with the above 
mentioned lower bounds on kernel size of traditional kernels on 
these classes. 
We complement our results by the following conditional lower bound. We
show that if a class~$\Cc$ is somewhere dense and closed under taking
subgraphs, then for some value of $r\in\N$ there cannot exist an
$\alpha$-approximate bi-kernel for the {\sc (Connected)
  \mbox{Distance-$r$} Dominating Set} problem on $\Cc$ for any
$\alpha > 1$ (assuming the Gap Exponential Time Hypothesis).  Observe
that a similar statement for subgraph-closed classes which contain
arbitrarily large bicliques also holds but is not too interesting, as
this contains as a subclass the class of all bipartite graphs on which
we can easily encode the {\sc Set Cover} problem.

\smallskip
\paragraph{Organization}

We explain our methods and provide a general framework for computing
lossy kernels for \textsc{Connected Dominating Set} in Section
\ref{sec:bicliquefree}.  Using this framework we can directly derive
the claimed bounds for $\alpha$-approximate kernels on biclique-free
graphs, while for bounded expansion and nowhere dense classes we
obtain $\alpha$-approximate bi-kernels for \textsc{Connected Distance-$r$
  Dominating Set} of polynomial size with the degree of the polynomial
depending both on $r$ and on $\alpha$. The rest of the paper is
devoted to the far more technical part of improving the bounds for
bounded expansion and nowhere dense classes of graphs.

\section{A general framework}\label{sec:general}

Although the technical details for dealing with biclique-free graphs
and with bounded expansion and nowhere dense classes are quite
different, the high-level approach is identical. The kernelization
algorithms follow the same two-step strategy.  First, our goal is to
compute a ``small'' set of vertices whose domination is sufficient,
i.e.\ the set of \emph{dominatees} or the so-called domination core.

\begin{definition}[$k$-domination core]
  Let $G$ be a graph and $Z \subseteq V(G)$. We say that $Z$ is a
  \emph{$k$-domination core} if every set $D$ of size at most $k$ that
  dominates $Z$ also dominates $V(G)$.
\end{definition}

Having found a domination core $Z$ of appropriate size, the next step
is to reduce the number of \emph{dominators}, i.e. vertices whose role
is to dominate other vertices, and the number of \emph{connectors},
i.e. vertices whose role is to connect the solution.

\begin{definition}
  Let $G$ be a graph and let $D,Z \subseteq V(G)$. We say that $D$ is
  a \emph{$Z$-dominator} if~$D$ dominates $Z$ in $G$, i.e. every
  vertex $z \in Z \setminus D$ is at distance at most one from some
  vertex in $D$. We denote by $\ds(G,Z)$ ($\cds(G,Z)$) the size of a
  smallest (connected) $Z$-dominator in~$G$.  By $\ds(G)$ ($\cds(G)$)
  we mean $\ds(G,V(G))$ ($\cds(G,V(G))$).
\end{definition}

We classify all vertices outside the core according to their
domination properties.

\begin{definition}\label{def:sim}
  Let $G$ be a graph and $Z\subseteq V(G)$. We define an equivalence
  relation $\sim_Z$ on $V(G)\setminus Z$ by
  $u\sim_Z v \Leftrightarrow N(u)\cap Z=N(v)\cap Z.$
\end{definition}

Clearly, to find a kernel for \textsc{Dominating Set} it is now
sufficient to construct the graph~$G'$ which contains the
$k$-domination core $Z$ and one representative of each equivalence
class of $\sim_Z$. Then $G$ admits a dominating set of size at most
$k$ if and only if $G'$ contains a $Z$-dominator of size at most
$k$. This simple two-step approach of computing a small domination
core $Z$ and then bounding the number of equivalence classes of the
relation $\sim_Z$ forms the basis of the kernelization algorithms for
\textsc{Dominating Set} in~\cite{DrangeDFKLPPRVS16,
  eickmeyer2016neighborhood,siebertz2016polynomial}. 
To control the number of classes of $\sim_Z$, we give the following
definition. The \emph{index} of an equivalence relation is the number of 
equivalence classes. 

\begin{definition}
Let $G$ be a graph. We define the \emph{neighborhood 
complexity function} of~$G$  as the function 
$\mu\colon \N\rightarrow \N$ with $\mu(z)=\max_{Z\subseteq V(G), |Z|=z}\textit{index }(\sim_Z)$. 
\end{definition}

For example all classes of bounded VC-dimension have 
polynomially bounded neighborhood complexity functions~\cite{sauer1972density, shelah1972combinatorial}, 
while bounded expansion classes have linear and nowhere dense classes have 
almost linear neighborhood
complexity~\cite{DBLP:journals/jcss/GajarskyHOORRVS17}. 

\begin{proposition}
  Let $\Cc$ be a class of graphs such that the neighborhood complexity
  function for all $G\in\Cc$ is bounded by a fixed polynomial and 
  such that on
  input $(G,k)$, for $G\in\Cc$, we can decide in polynomial time whether
  $\ds(G)>k$ or otherwise compute a $k$-domination core
  $Z\subseteq V(G)$ of size polynomial in $k$. Then \textsc{Dominating
    Set} parameterized by~$k$ admits a polynomial-size bi-kernel on
  $\Cc$.
\end{proposition}

The above proposition can be applied, e.g., to obtain a polynomial
kernel on biclique-free graphs (we will prove the existence of a
polynomial $k$-domination core below).  However, as the hardness
results even for degenerate graphs show, this approach does not extend
to connected dominating sets. We may have to include more vertices in
the kernel to ensure connectivity of the dominating sets. This
turns out to be a major problem for the construction of polynomial
size kernels for \textsc{Connected Dominating Set}.

\pagebreak

When reducing the number of vertices outside the domination core, we borrow approximation 
techniques that are closely related to the {\sc Steiner Tree} problem.

\begin{definition}
  Let $G$ be a graph and let $Y\subseteq V(G)$ be a set of
  \emph{terminals}.  A \emph{Steiner tree} for~$Y$ is a subtree of $G$
  spanning $Y$.  We write $\mathbf{st}_G(Y)$ for the order of (i.e.\
  the number of vertices of) the smallest Steiner tree for $Y$ in $G$
  (including the vertices of $Y$).  If $\YYY=\{V_1,\ldots, V_t\}$ is a
  family of vertex disjoint subsets of $G$, a \emph{group Steiner
    tree} for $\YYY$ is a subtree of $G$ that contains (at least) one
  vertex of each group~$V_i$. We write $\mathbf{st}_G(\YYY)$ for the
  order of the smallest group Steiner tree for $\YYY$.
\end{definition}

\noindent 
The {\sc Group Steiner Tree} problem on $t$ groups can be solved in
$\Oh(2^t \cdot n^{\Oh(1)})$-time~\cite{MisraPRSS10}.
The following definition and proposition form the key to our approach to
handle connectivity in the lossy kernelization framework.

\begin{definition}
  Let $D$ be a connected graph and $t\in {\mathbb N}$.  A
  \emph{$(D,t)$-covering family} is a family $\mathcal{F}(D,t)$ of
  connected subgraphs of $D$ such that $(i)$ for each
  $T \in {\mathcal{F}(D,t)}$, $\vert V(T) \vert \leq 2t$ and $(ii)$
  $\bigcup_{T \in \mathcal{F}(D,t)}{V(T)} = V(D)$.
\end{definition}

\begin{proposition}\label{prop-covering-family}
  Let $D$ be a connected graph and $t\in {\mathbb N}$.  Then there is
  a $(D,t)$-covering family $\mathcal{F}(D,t)$ such that
  $|\mathcal{F}(D,t)| \leq \frac{|V(D)|}{t} + 1$, and
  $\sum_{T \in \mathcal{F}(D,t)}{|V(T)|} \leq (1 + \frac{1}{t}) |V(D)|
  + 1$.
\end{proposition}

\begin{proof}
  Let $T_D$ be a spanning tree of $D$.  We create a $(D,t)$-covering
  family $\mathcal{F}(D,t) = \{T_1, T_2, \ldots, T_\ell\}$, which is a
  set of subtrees of $T_D$ constructed as follows. We root $T_D$ at an
  arbitrary vertex $r\in V(T_D)$. For any pair of vertices
  $u,v\in V(T_D)$, $u$ is called a child of~$v$ if $uv\in E(T_D)$ and
  $v$ lies on the path from $u$ to $r$. For each vertex
  $v \in V(T_D)$, we let
  $\text{weight}(v) = 1 + \sum_{u \in \text{child}(v)}
  \text{weight}(u)$,
  where $\text{child}(v)$ denotes the set of children of $v$ in $T_D$.
  In other words, $\text{weight}(v)$ is the number of vertices in the
  subtree rooted at $v$.  Leaves have weight one.  We use $T_v$ to
  denote the subtree rooted at $v$. We construct
  $\mathcal{F}(D,t) = \{T_1, T_2, \ldots\}$ from $T_D$ as follows:
  \medskip
  \begin{itemize}
  \item[(1)] If $T_D$ is empty, terminate. Otherwise, compute the
    weights of all vertices in $T_D$, then sort the vertices in
    increasing order of weight.
  \item[(2)] If there exists a vertex whose weight is between $t$ and
    $2t$ (inclusive), pick the vertex with the smallest such weight,
    add $T_v$ to $\mathcal{F}(D,t)$, delete $T_v$ from $T_D$, then go
    back to step (1).
  \item[(3)] If there exists a vertex whose weight is strictly greater
    than $2t$, pick the vertex with the smallest such weight, greedily
    compute a subset $S \subseteq \text{child}(v)$ such that
    $t < \sum_{u \in S} \text{weight}(u) < 2t$, let
    $R = \text{child}(v) \setminus S$, add
    $T_v - \bigcup_{w \in R} V(T_w)$ to $\mathcal{F}(D,t)$, delete
    $T_u$ from $T_D$, for every $u \in S$, then go back to step (1).
    Note that by our choice of $v$, all children of $v$ must have
    weight at most $t - 1$ as otherwise case (2) would apply.
  \item[(4)] Otherwise, every vertex in $T_D$ has weight strictly less
    than $t$ and hence $T_D$ has at most $t$ vertices (by the
    definition of the weight function). In this case, simply add $T_D$
    to $\mathcal{F}(D,t)$ and terminate.
\end{itemize}

\medskip
It is not hard to see that whenever we delete a subtree in the above
procedure we never disconnect the tree.  Moreover, the procedure
terminates only when the tree becomes empty and hence
$\bigcup_{T \in \mathcal{F}(D,t)}{V(T)} = V(T_D)$.  Whenever a tree is
added to $\mathcal{F}(D,t)$ in either step (2) or (3), the size of
that tree is at least equal to $t$ and at most equal to $2t$.  For
step (4), the size of the added tree is strictly less than
$t$. Combining those two facts, we know that $|\mathcal{F}(D,t)|$ is
at most $\frac{|V(D)|}{t} + 1$, as the size of the current tree $T_D$
is reduced by at least $t$ for all but one subtrees.  To prove the
last inequality, i.e.
$\sum_{T \in \mathcal{F}(D,t)}{|V(T)|} \leq (1 + \frac{1}{t}) |V(D)| +
1$,
we let $\text{mult}(v)$ denote the multiplicity of~$v$ minus one,
i.e. the number of subtrees $T \in \mathcal{F}(D,t)$ in which $v$
appears minus one.  So if a vertex occurs only once, its multiplicity
is zero. 

Observe that
$\sum_{T \in \mathcal{F}(D,t)}{|V(T)|} \leq |V(D)| + \sum_{v \in
  V(T_D)}{\text{mult}(v)}$.
However, whenever the multiplicity of a vertex increases by one, the
size of the running subtree decreases by at least~$t$ (step (3)) or
the procedure terminates (step (4)).  Hence,
$\sum_{v \in T_D}{\text{mult}(v)} \leq \frac{|V(D)|}{t} + 1$, as needed. 
\end{proof}

While the previous proposition allows us to ``take apart'' connected dominating sets, the next
proposition explains how to ``put them back together''.

\begin{proposition}
\label{approx-cds-by-ds}
Let $G$ be a graph, $X\subseteq V(G)$, such that $G[X]$ is connected,
and let $D$ be an $X$-dominator such that $G[D]$ has at most $p$
connected components.  Then a set $Q \subseteq X$ of size at most $2p$
such that $G[D \cup Q]$ is connected, can be computed in polynomial
time.
\end{proposition}

\begin{proof}
If $D$ is connected, then we set $Q=\emptyset$ and the proposition follows. Hence, assume that $G[D]$ contains at least $2$ connected components.
Now, the algorithm for finding $Q$ does the following: 
\medskip
\begin{itemize}
\item[(1)] We start by setting $Q=\emptyset$. 
\item[(2)] If $G[D\cup Q]$ is connected, output $Q$.
\item[(3)] If there is a vertex $z$ in $X\setminus (D\cup Q)$ that is dominated from at least $2$ different components of $G[D\cup Q]$, we add $z$ into $Q$. 
\item[(4)] If there is an edge $uv$ with both end vertices in $X$ such that the number of components of $G[D \cup Q]$ that dominate at least one of $\{u,v\}$ is at least 2, we 
add $u$ and $v$ into $Q$. 
\end{itemize}
\medskip

Note that in both steps (3) and (4) we increase the size of $Q$ by at most two and decrease the number 
of connected components in $G[D\cup Q]$ by at least one. Hence after at most $p-1$ steps we end up with only one connected component. 
Now we show that if we cannot apply  step (3) nor step (4), then $G[D\cup Q]$ is connected. Since, we cannot apply step (3), each vertex of $X$ is dominated by 
exactly one connected component of $D\cup Q$. Moreover, since we cannot apply  step (4), if $uv$ is an edge 
in $X$, then both $u$ and $v$ are dominated by the same component of $G[D\cup Q]$. 
Finally, since $X$ is connected, there is a path between every pair of vertices of $X$ and all vertices 
on that path have to be dominated by the same distinct component of $G[D\cup Q]$, which implies that $G[D\cup Q]$ consists of a single connected component. 
\end{proof}

We have now collected all the tools required to control the number of
vertices which have to be added to ensure connectivity. 

\begin{theorem}\label{thm:generalkernel}
Let $\Cc$ be a class of graphs such that the neighborhood complexity
  function for all $G\in\Cc$ is bounded by a fixed polynomial of
  degree $d$ and 
   such that on input $(G,k)$, for $G\in\Cc$, we can decide in polynomial time
  whether $\cds(G)>k$ or otherwise compute a $k$-domination core
  $Z\subseteq V(G)$.  Then for every $\epsilon>0$, \textsc{Connected
    Dominating Set} parameterized by $k$ admits a
  $(1+\epsilon)$-approximate kernel with $|Z|^{\Oh(d/\epsilon)}$ vertices
  on $\Cc$.
\end{theorem}

The rest of the section is mainly devoted to prove Theorem~\ref{thm:generalkernel}. 

\medskip
\subparagraph{The reduction algorithm} Let $(G,k)$ be the input
instance, where $G\in\Cc$ is connected and~$k$ is a positive
integer. We first describe the reduction algorithm
${\cal R}_{\cal A}$.  As a first step we run the polynomial time
algorithm (which exists by assumption of the theorem) to decide
whether $\cds(G) > k$ and otherwise compute a $k$-domination core
$Z\subseteq V(G)$. In the first case, we output a trivial negative
instance $((\{v\},\emptyset),0)$. In the second case, we proceed as
follows.

We partition the graph into two sets $Z$ and $R = V(G) \setminus Z$.
We compute the equivalence relation $\sim_Z$ on $R$ (see
Definition~\ref{def:sim}), that is, we partition vertices in $R$
according to their neighborhoods in $Z$. This is clearly possible in
polynomial time.  Let $\mathcal{R}$ be the set of equivalence classes
defined by $\sim_Z$.  As a direct implication of
our assumption, we can bound the size of $\mathcal{R}$ by
$\Oh(|Z|^d)$.

\begin{proposition}\label{cl:num-classes}
  The  relation $\sim_Z$ has $\Oh(|Z|^d)$ classes,
  that is, $|\mathcal{R}| \in \Oh(|Z|^d)$.
\end{proposition}

As $Z$ is a $k$-domination core, to find a dominating set of size at
most $k$ it is enough to find a set which dominates $Z$.  Hence for
the purpose of domination, it is redundant to pick more than one
vertex from an equivalence class in $\mathcal{R}$.  The following
construction finds a small set of relevant vertices which
``approximately'' preserves the connectivity requirements.

Let $t \geq 1$ be a constant, which we fix later.  Let $\mathcal{Z}$
be the family of groups $\{\{z\}~|~z\in Z\}$ and let $\mathcal{R}$ be
the set of equivalence classes defined by $\sim_Z$.  The set
$\mathcal{R} \cup \mathcal{Z}$ forms a family of groups of vertices in
$V(G)$.  For every subset
$\mathcal{Q} = \{Q_1, \ldots, Q_{\ell}\} \subseteq \mathcal{R} \cup
\mathcal{Z}$
of size at most $2t$ of groups in $\mathcal{R} \cup \mathcal{Z}$,
construct a {\sc Group Steiner Tree} instance on the graph $G$ with
groups $Q_1, \dots, Q_{\ell}$.  Note that since $t$ is a constant each
instance can be solved in polynomial time using the algorithm of Misra
et al.~\cite{MisraPRSS10}. For each subset $\mathcal{Q}$ denote by
$T_{\mathcal{Q}}$ the corresponding solution.  For every instance that
we solve, if the size of $T_{\mathcal{Q}}$ is at most $2t$ then we
mark the vertices of $T_{\mathcal{Q}}$ in~$G$. We denote the set of
all marked vertices by $\bigcup T_{\mathcal{Q}}$.  If
$\bigcup T_{\mathcal{Q}}$ is not a dominating set in $G$, then we may
declare that $\cds(G) > k$.  Otherwise, since $G$ is assumed to be
connected, we can run the polynomial-time algorithm of
Proposition~\ref{approx-cds-by-ds} (with parameter $X=V(G)$) to obtain
a set $W \subseteq V(G)$ such that $\bigcup T_{\mathcal{Q}} \cup W$ is
a connected dominating set in~$G$ and
$|\bigcup T_{\mathcal{Q}} \cup W| \leq 3|\bigcup T_{\mathcal{Q}}|$.
Let $Y=\bigcup T_{\mathcal{Q}} \cup W$. We output the instance
$(G[Y], k)$.

\subparagraph*{Approximation guarantee} Now we prove that
$\text{OPT}(G[Y],k)\leq (1+\epsilon)\text{OPT}(G,k)$.  Let~$D^*$ be a
connected dominating set of $G$ of minimum cardinality.  If
$\vert D^*\vert >k$, then
$\text{OPT}(G[Y],k)\leq (1+\epsilon)\text{OPT}(G,k)$ holds trivially.
So assume that $\vert D^* \vert \leq k$.  We let
$\mathcal{F}(D^*,t) = \{T_1, T_2, \cdots, T_m\}$ denote a
$(D^*,t)$-covering family. Proposition~\ref{prop-covering-family}
implies that there exists such a family for which
$|\mathcal{F}(D^*,t)| \leq \frac{|V(D^*)|}{t} + 1$ and
$\sum_{T \in \mathcal{F}(D^*,t)}{|V(T)|} \leq (1 + \frac{1}{t})
|V(D^*)| + 1$.
Moreover, the size of each connected subgraph $T$ (in this case also
subtree) is at most $2t$.
We construct a new family $\mathcal{F'}$ from $\mathcal{F}(D^*,t)$ as
follows.  For each $T \in \mathcal{F}(D^*,t)$, we replace $T$
by~$T_{\mathcal{Q}}$, where ${\mathcal Q}$ is the set of groups from
${\mathcal R}\cup \mathcal{Z}$ such that $Q\in {\mathcal Q}$ if and
only if $V(T)\cap Q\neq \emptyset$ and $T_{\mathcal{Q}}$ is the set of
marked vertices in an optimal Steiner tree connecting vertices from
the groups in $\mathcal{Q}$.  Note that the fact that $T$ is of size
at most $2t$ guarantees the existence of $T_{\mathcal{Q}}$ (by
construction).  Moreover, the size of $T_{\mathcal{Q}}$ is at most the
size of $T$, since $T$ is also a solution for {\sc Group Steiner Tree}
for $\mathcal{Q}$.  Let $D_{\mathcal{F'}}$ denote the union of all
vertices in~$\mathcal{F}'$.

Let $D'$ be a subset of $D_{\mathcal{F'}}$, of cardinality at most
$\vert D^*\vert$, such that for any $w \in D^*$, there is a vertex
$w' \in D'$ with the property that
$\{w,w'\}\subseteq Q \in {\mathcal R} \cup \mathcal{Z}$ and
$w' \in D_{\mathcal{F'}}$.  That is, if $D^*$ has a vertex from a
group $Q$ in ${\mathcal R}\cup \mathcal{Z}$, then $D'$ also has a
vertex from group $Q$.  We claim that~$D'$ is a dominating set in
$G$. Notice that $Z\cap V(D) = Z \cap D'$ and if any vertex in~$Z$ is
adjacent to a vertex in a group $Q$, then it is adjacent to all
vertices in group $Q$. This implies that $D'$ also dominates $Z$ and
since $\vert D'\vert \leq \vert D \vert \leq k$, by the definition of
a $k$-domination core,~$D'$ is a dominating set in $G$.

This implies that $D_{\mathcal{F'}}\supseteq D'$ is also a dominating
set in $G$.  Applying Proposition~\ref{approx-cds-by-ds} in $G[Y]$
(with $D_{\mathcal{F'}}$ as dominator and since $G[Y]$ is connected),
we obtain a connected dominating set of size at most
$2|\mathcal{F}(D^*,t)|+|D_{\mathcal{F'}}|\le \frac{2 \vert V(D^*)
  \vert}{t} + 2 + (1 + \frac{1}{t})\vert V(D^*) \vert + 1 = (1 +
\frac{3}{t})\vert V(D^*)\vert + 3$.
Now we can fix the constant $t$ appropriately (as roughly
$\frac{3}{\epsilon}$) and we get that
$\text{OPT}(G[Y],k)\leq (1+\epsilon)\text{OPT}(G,k)$.

\medskip
\subparagraph*{Size of the kernel} Now we show that
$\vert Y \vert \in |Z|^{\Oh(d)/\epsilon}$. By
Proposition~\ref{cl:num-classes}, we have that
$\vert {\mathcal R}\cup {\mathcal Z}\vert = \Oh(Z^{d})$. From the
construction, it follows that
$\vert \bigcup T_{\mathcal{Q}}\vert= \Oh(2t\vert {\mathcal R}\cup
{\mathcal Z}\vert^{\Oh(t)}) = \vert Z\vert^{\Oh(d/\epsilon)}$.
Notice that $Y= \bigcup T_{\mathcal{Q}} \cup W$, where $W$ is obtained
by applying Proposition~\ref{approx-cds-by-ds} and hence we have that
$Y=|\bigcup T_{\mathcal{Q}} \cup W| \leq 3|\bigcup T_{\mathcal{Q}}|=
\vert Z\vert^{\Oh(d/\epsilon)}$.

\medskip
\subparagraph*{The solution lifting algorithm} The solution lifting
algorithm works as follows. Given a solution $D'$ to the reduced
instance $(G',k')$, if $D'$ is not a connected dominating set of $G'$,
then the solution lifting algorithm will output $\emptyset$.  If $D'$
is a connected dominating set, then the algorithm returns $D'$ if
$|D'|\le k$ and $V(G)$ otherwise.  Let $D$ be the output of the
solution lifting algorithm.

\medskip
\subparagraph*{The final step}

We prove that the above reduction algorithm together with the solution
lifting algorithm constitute a $(1+\epsilon)$-approximate kernel. Note
that if $D'$ is not a valid solution of $G'$, then $\emptyset$ is not
a valid solution for $G$ and $\Cds(G',k',D') = \Cds(G,k,D)=\infty$. Hence
we can restrict ourselves to the case when $D'$ is a connected
dominating set of~$G'$.  First, consider the case where the reduction
algorithm outputs $Y\subseteq V(G)$ and the reduced instance is hence
$(G',k')=(G[Y],k)$. From our above observation, we have that
$\Opt(G[Y],k)\leq (1+\epsilon)\Opt(G,k)$.  We show that in this case
$\Cds(G,k,D) = \Cds(G',k',D')$. If $|D'|>k$, then
$\Cds(G,k,D)=\Cds(G,k,V(G)) = k+1=\Cds(G',k',D')$.  So assume that
$|D'|\leq k$, which implies $D=D'$. Since $D'$ is a connected
dominating set of $G[Y]$ and $Y$ contains a $k$-domination core of
$G$, it follows that $D'$ dominates $G$ and
$\Cds(G,k,D) = \Cds(G',k',D')$.  Combining $\Cds(G,k,D) = \Cds(G',k',D')$
and $\Opt(G[Y],k)\leq (1+\epsilon)\Opt(G,k)$ we get
$\frac{\Cds(G,k,D)}{\Opt(G,k)}\leq
(1+\epsilon)\frac{\Cds(G',k',D')}{\Opt(G',k')}.$
When $(G',k')=((\{v\},\emptyset),0)$, we can easily verify that the
above mentioned approximation guarantee holds.  \hfill $\square$

\bigskip The remainder of the paper is concerned with proving the
existence of small domination cores for concrete sparse classes of
graphs. 
The most technical part will be to prove the existence of a
linear domination core for bounded expansion classes (the definition
of a domination core is slightly changed to obtain such good bounds,
see Section~\ref{sec:exp}). Most surprisingly, the general framework summarized 
in Theorem~\ref{thm:generalkernel} does not produce a
bi-kernel of size $\Oh(|Z|^{1/\epsilon})$ but rather of size
$f(\epsilon)\cdot |Z|$ for some function $f$ on bounded expansion
classes and of polynomial size on nowhere dense classes.

\section{Biclique-free graphs}\label{sec:bicliquefree}
In this section we fix a class $\Cc$ which excludes
some biclique $K_{d,d}$ for a fixed positive integer~$d$. 
All of our arguments can be easily extended to $K_{i,j}$-free graphs, 
for positive integers $i$ and $j$, but we use $K_{d,d}$-freeness 
for simplicity. 
We show that {\sc Connected Dominating Set}, parameterized 
by solution size, admits a PSAKS on $K_{d,d}$-free graphs. More precisely, we show that
for every $\epsilon>0$, \textsc{Connected Dominating Set} admits 
a $(1+\epsilon)$-approximate kernel on at most $k^{\Oh(d^2/\epsilon)}$  vertices. 
According to the framework presented in the previous section
we want to prove that $K_{d,d}$-free graphs admit small
domination cores and have polynomially-bounded neighborhood complexity. 
The latter is well known, as the next lemma shows.

\begin{lemma}[\cite{LokshtanovMPRS15}]\label{kdd}
Let $G$ be a $K_{d,d}$-free graph and let $Z\subseteq V(G)$. 
Then the relation $\sim_Z$ has at most $2d\cdot |Z|^d$ equivalence
classes. 
\end{lemma}	

Our algorithm to find a $k$-domination core heavily relies on 
several earlier results. We combine
several of the ideas introduced in~\cite{alon2009linear,
DawarK09,DrangeDFKLPPRVS16,
philip2009solving,PhilipRS12,telle2012fpt} and, as a byproduct, are
able to simplify the proofs of some of the results. 
Note that $V(G)$ is a $k$-domination core. 
We initialize $Z = V(G)$ and repeatedly 
reduce the size of $Z$, vertex by vertex, 
maintaining a $k$-domination core throughout the process. 

\begin{lemma}\label{cdcore}
There exists a polynomial-time algorithm that, given a $K_{d,d}$-free graph $G$ and a $k$-domination core $Z \subseteq V(G)$ with $|Z| > (2d + 1)k^{d + 1}$,  
either correctly concludes that $\ds(G) > k$ (and hence $\cds(G) > k$) or finds a 
vertex $z \in Z$ such that $Z \setminus \{z\}$ is a $k$-domination core. 
\end{lemma}

\begin{proof}
We design such an algorithm as follows.  
If there is no vertex $v \in V(G)$ such that the cardinality of the neighborhood of $v$ in $Z$ is at least $\lceil \frac{\vert Z\vert}{k} \rceil$, 
then the algorithm terminates and declares that $\ds(G) > k$.  
Otherwise to find $z \in Z$, the algorithm constructs a sequence of sets $Z=X_0 \supseteq X_1 \supseteq \dots \supseteq X_{\ell}$ and a set 
$S = \{v_1, \ldots, v_{\ell}\} \subseteq V(G)$ such that $X_i \subseteq N[v_i]$. 
We construct the sets $Z= X_0 \supseteq X_1 \supseteq \dots \supseteq X_{\ell}$ and the set $S$ using an iterative procedure. 
Initially, we set $S := \emptyset$ and $X_0:=Z$. 
At step $i$, if there is a vertex $v_i \in V(G)\setminus S$ whose neighborhood in $X_{i-1}$ has at least $\lceil \frac{|X_{i-1}|}{k} \rceil$ vertices, 
then add $v_i$ to~$S$ and set $X_i:=X_{i-1}\cap N[v_i]$. 
If no such vertex $v_i$ exists, then the algorithm outputs
an arbitrary vertex $z\in X_{i-1}\setminus S$ and stops.

The above procedure will construct a sequence  $Z=X_0 \supseteq X_1 \supseteq \dots \supseteq X_{\ell}$ and a set 
$S=\{v_1,\ldots,v_{\ell}\}$. We first claim that $\ell<d$. Suppose $\ell\geq d$. Then 
there is a complete bipartite subgraph $H$ of $G$ with bipartition $\{v_1,\ldots,v_d\}$ and $X_{d}\setminus \{v_1,\ldots,v_d\}$. 
Since $\vert X_{d}\vert \geq \frac{\vert Z\vert}{k^d} \geq \frac{(2d + 1)k^{d + 1}}{k^d} \geq 2d+1$, $H$ contains $K_{d,d}$ as a subgraph, which is a contradiction 
to the fact that~$G$ is $K_{d,d}$-free.  Hence $\ell<d$. Moreover, since $\vert X_{\ell}\vert \geq \frac{\vert Z\vert }{(k)^\ell} \geq 
(2d+1)k^2$, there always exists a vertex $z\in X_{\ell}\setminus S$ that the algorithm can select to output.

Now we prove the correctness of the algorithm. Clearly, if   
there is no vertex $v\in V(G)$ such that the cardinality of the neighborhood of $v$ in $Z$ is at least $\lceil \frac{\vert Z\vert}{k} \rceil$, 
then $\ds(G) > k$ and the algorithm declares it correctly. 
Otherwise let $z\in X_{\ell}\setminus S$ be the output and let $Z'=Z\setminus\{z\}$.
We need to prove that $Z'$ is still a $k$-domination core. 
Let $D$ be a set of size at most $k$ that dominates $Z'$. 
To prove that $D$ is a dominating set in $G$, it is enough to show that $D$ dominates 
$Z$ (because $Z$ is a $k$-domination core).  Since $D$ already dominates $Z'=Z\setminus \{z\}$, it suffices to show that $D$ dominates $z$. 
Notice that every vertex in $S$ is adjacent to $z$. 
Hence, to show that $D$ also dominates $Z$, it is enough to show that 
$D \cap S \neq \emptyset$. We will show that if $D \cap S = \emptyset$, then $D$ cannot dominate all of $X_\ell\setminus\{z\}\subseteq Z'$.
Since our algorithm stops at step $\ell+1$, we know that every vertex outside $S$ 
can dominate strictly less than $(2d + 1)k^{d-\ell}$ vertices from $X_{\ell}$. 
As $|X_{\ell} \setminus \{z\}| \geq (2d + 1)k^{d+1-\ell}$, 
$\vert D\vert \leq k$
and every 
vertex outside $S$ can dominate at most $(2d + 1)k^{d - \ell} - 1$ vertices of $X_{\ell}$, we have that $D \cap S \neq \emptyset$. 
This completes the proof of the lemma. 
\end{proof}	

\medskip
The next theorem follows immediately from Theorem~\ref{thm:generalkernel} and Lemmata~\ref{kdd} and~\ref{cdcore}. 

\medskip
\begin{theorem}\label{thkddfree}
For every $\epsilon > 0$, {\sc Connected Dominating Set}, parameterized by solution size,  
admits a $(1 + \epsilon)$-approximate kernel with $k^{\Oh(d^2/\epsilon)}$ vertices on $K_{d,d}$-free graphs. 
\end{theorem}  

\vfill\pagebreak

\section{Bounded expansion graphs}\label{sec:exp}

In this section we show that {\sc Connected Dominating Set}, parameterized by solution size, admits a 
$(1+\epsilon)$-approximate bi-kernel on at most $\Oh(f(\epsilon) \cdot k)$ vertices. 
The reduced instance will be an instance of {\sc Subset Connected Dominating Set} (SCDS), defined as follows: 

\[\text{SCDS}((G,S),k,D) = \left\{
  \begin{array}{rl}
    \infty & \text{if $D$ is not a connected}\\
    & \hspace{0.5cm} \text{$S$-dominator in $G$}\\
    \text{min}\{|D|, k + 1\} & \text{otherwise}
  \end{array}
\right.
\]

The first phase of our algorithm, i.e. finding a domination core, closely follows the work of Drange et al.~\cite{DrangeDFKLPPRVS16} but requires subtle changes. 
We fix a graph class $\mathcal{G}$ that has bounded expansion and let $(G,k)$ be the input 
instance of CDS, where $G\in \mathcal{G}$ and $G$ is connected. 

\subsection{Preliminaries}  
We note here that many well-known sparse graph classes such as planar graphs, 
graph classes with bounded genus, treewidth, or
degree as well as graph classes characterizable by finite set of 
forbidden minors have all bounded expansion.  
We refer the reader to~\cite{CLASSES,NO10} for more details. 

\begin{definition}[Shallow minors]
A graph $M$ is an \emph{r-shallow minor} of $G$, for some $r\in \N$, if there exists a family of 
disjoint subsets $V_1, \ldots, V_{|M|}$ of $V(G)$ such that: 
\medskip
\begin{itemize}
\item[(1)] each graph $G[V_i]$ is connected and has radius at most $r$, and 
\item[(2)] there is a bijection $\omega \colon V(M) \rightarrow \{V_1, \ldots, V_{|M|}\}$ such that 
for every edge $uv \in E(M)$ there is an edge in $G$ with one endpoint in $\omega(u)$ and another in $\omega(v)$. 
\end{itemize}
\end{definition}

The set of all $r$-shallow minors of a graph $G$ is denoted by $G \triangledown r$. 
The set of all $r$-shallow minors of all members of a graph class $\mathcal{G}$ is denoted by 
$\mathcal{G} \triangledown r = \bigcup_{G \in \mathcal{G}}(G \triangledown r)$. 
In the introduction we defined bounded expansion classes by bounding the
edge density of shallow topological minors. It will be convenient to work with the following
equivalent definition of bounded expansion classes. 

\begin{definition}[Grad and bounded expansion~\cite{nevsetvril2008grad}]
For a graph $G$ and an integer $r \geq 0$, the \emph{greatest reduced average density (grad) at depth r} 
is, 
$
\nabla_r(G) = \text{max}_{M \in G \triangledown r} \text{density}(M) = \text{max}_{M \in G \triangledown r} |E(M)|/|V(M)|.
$
For a graph class $\mathcal{G}$, we let $\nabla_r(\mathcal{G}) = \text{sup}_{G \in \mathcal{G}} \nabla_r(G)$. A graph class 
$\mathcal{G}$ has \emph{bounded expansion} if there is a function $f \colon \mathbb{N} \rightarrow \mathbb{R}$ such that 
for all $r$ we have $\nabla_r(\mathcal{G}) \leq f(r)$. 
\end{definition}

\noindent For ease of presentation, and since we deal with several constants, we shall use the following convention. 
We assume that a graph class of bounded expansion $\mathcal{G}$ is fixed, and hence so 
are the values of $\nabla_i(\mathcal{G})$, for all non-negative integers $i$. 
This assumption is not strictly required but it significantly simplifies the analysis. 
Next, we state some useful properties of bounded expansion graphs that will be used later on. 

\begin{lemma}[\cite{DrangeDFKLPPRVS16}]\label{lem-uvertex}
Let $G$ be a graph and let $G'$ be the graph obtained by adding a universal vertex to $G$, 
i.e. a vertex adjacent to all vertices in $G$. Then, for any non-negative integer $r$, $\nabla_r(G') \leq \nabla_r(G) + 1$.
\end{lemma}

\noindent Given two graphs $G$ and $H$, the \emph{lexicographic product} $G \odot H$ is defined 
as the graph on the vertex set $V(G) \times V(H)$ where vertices $(u,a)$ and $(v,b)$ 
are adjacent if $uv \in E(G)$ or if $u = v$ and $ab \in E(H)$. 

\begin{lemma}[\cite{Har-PeledQuanrud15ar,Har-PeledQuanrud15}]\label{lem-lex} 
For a graph $G$ and non-negative integers $t \geq 1$ and $r$ we have 
$\nabla_r(G \odot K_t) \leq 5t^2(r+1)^2 \cdot \nabla_r(G)$. 
\end{lemma}

Let $G$ be a graph and $X$ be a subset of its vertices. For $u \in V(G) \setminus X$, we define the \emph{$r$-projection} 
of $u$ onto $X$ as follows: $M_r^G(u,X)$ is the set of all vertices $w \in X$ for which there exists a path in $G$ 
that starts in $u$, ends in $w$, has length at most $r$, and whose internal vertices do not belong to $X$. 
Note that $M_1^G(u,X)=N_X(u)$. 
We omit the superscript when the graph is clear from context. 

\begin{lemma}[\cite{DrangeDFKLPPRVS16}]\label{lem-closure}
Let $\mathcal{G}$ be a class of graphs of bounded expansion. 
There exists a polynomial-time algorithm that, given a graph $G \in \mathcal{G}$, $X \subseteq V(G)$, and 
an integer $r \geq 1$, computes the \emph{$r$-closure} of $X$, denoted by $\cl_r(X)$, with the following properties: 
\begin{itemize}
\item[(1)] $X \subseteq \cl_r(X) \subseteq V(G)$, 
\item[(2)] $|\cl_r(X)| \leq C_{\textbf{cl1}} \cdot |X|$, and 
\item[(3)] $|M_r^G(u, \cl_r(X))| \leq C_{\textbf{cl2}}$ for each $u \in V(G) \setminus \cl_r(X)$, where $C_{\textbf{cl1}}$ and $C_{\textbf{cl2}}$ 
are constants depending only on $r$ and a fixed (finite) number of grads of $\mathcal{G}$.  
\end{itemize}
\end{lemma}

\begin{lemma}[\cite{DrangeDFKLPPRVS16,DBLP:journals/jcss/GajarskyHOORRVS17,DBLP:journals/talg/0002LPRRSS16}]\label{lem-twin-expansion}
Let $\mathcal{G}$ be a class of graphs of bounded expansion and $r\in {\mathbb N}$. Let $G \in \mathcal{G}$ 
be a graph and let $X \subseteq V(G)$. Then 
\[|\{Y \mid Y = M_r(u,X)~\text{for $u \in V(G) \setminus X$}\}| \leq C_{\textbf{ex}} \cdot |X|,\] where 
$C_{\textbf{ex}}$ is a constant depending only on $r$ and a fixed (finite) number of grads of $\mathcal{G}$.  
\end{lemma}

\begin{lemma}[\cite{DrangeDFKLPPRVS16}]\label{lem-charging}
Let $G$ be a bipartite graph with bipartition $(X,Y)$ that belongs to some graph class $\mathcal{G}$ such 
that $\nabla_1(\mathcal{G}) \geq 1$. Moreover, suppose that for every $u \in Y$ we have that $N(u) \neq \emptyset$, 
and that for every distinct $u_1,u_2 \in Y$ we have $N(u_1) \neq N(u_2)$, i.e. $Y$ is twin-free. 
Then there exists a mapping $\phi \colon Y \rightarrow X$ with the 
following properties:
\begin{itemize}
\item[(1)] $u\phi(u) \in E(G)$ for each $u \in Y$ and 
\item[(2)] $\phi^{-1}(v) \leq C_{\textbf{ch}}$ for each $v \in X$, 
where $C_{\textbf{ch}}$ is a constant depending only on a fixed (finite) number of grads of $\mathcal{G}$.
\end{itemize}
\end{lemma}

\subsection{Finding the domination core}

To obtain a linear domination core for classes of bounded expansion
we have to invest a considerable amount of work. The following
construction shows that we cannot work with $k$-domination cores. 

\begin{lemma}
There exists a class $\Cc$ of bounded expansion such that for all 
$k\in \N$ there is $G\in\Cc$ such that every $k$-domination 
core for $G$ has $\Omega(k^2)$ vertices. 
\end{lemma}

\begin{proof}
Let $G_{k,m}$ be the graph which consists of a $k\times m$ grid
where additionally each row is dominated by an apex vertex. Formally, 
$V(G_{k,m})=\{v_{ij}~|~1\leq i\leq k,1\leq j\leq m\}\cup \{w_i~|~1\leq i\leq k\}$ and $v_{ij}v_{i'j'}\in E(G)$ if $|i-i'|+|j-j'|=1$ and 
$w_iv_{i,j}\in E(G)$ for all $i,j$. We claim that any 
$(2k-2)$-domination core of $G_{k,m}$ must contain at least $k^2$
vertices, where $m\geq 2k$. Assume towards a contradiction that $Z$ is a domination core
with less than $k^2$ vertices. This implies that one row, say row $i$, 
of the grid
contains less than $k$ vertices. Denote these vertices by
$z_1,\ldots, z_\ell$ for $\ell<k$. We show that there exists a 
$Z$-dominator of size $2k-2$ which does not dominate the whole
graph, contradicting the assumption that $Z$ is a $(2k-2)$-domination
core. Let $D$ be the set $\{w_j~|~j\neq i\}\cup \{z_1,\ldots,z_\ell\}$,
which clearly dominates $Z$ and has size at most $2k-2$, but 
does not dominate $G$. It is also not difficult to see that the
class $\{G_{k,m}~|~k,m\in\N\}$ has bounded expansion. 
\end{proof}

For this reason, we work with the following notion of a $c$-exchange domination core, which is different from the definition 
used in the previous section and from the one considered in~\cite{DrangeDFKLPPRVS16}. Here, $c$ is a fixed constant which we set later. 

\begin{definition}[$c$-exchange domination core]\label{def-dom-core-expansion}
Let $G$ be a graph and $Z$ be a subset of vertices of $G$. 
We say that $Z$ is a \emph{$c$-exchange domination core} if for every set $X$ that 
dominates~$Z$ one of the following conditions holds: (1.) $X$ dominates $G$, or (2.) there exist $A \subseteq X$ and $B \subseteq V(G)$ such that $|B| < |A| \le c$ and 
$(X \setminus A) \cup B$ is a set that dominates~$Z$. 
Moreover the number of connected components of $(X \setminus A) \cup B$ is at most the number of connected components of $X$. 
In particular, if $X$ is a connected set then $(X \setminus A) \cup B$ is also connected.
\end{definition}

\begin{proposition}\label{prop-dom-core}
Let $G$ be a graph, $c$ be a constant, and $Z$ be a $c$-exchange domination core of $G$, and $X \subseteq V(G)$ be a $Z$-dominator.  
Then, there is a set $Y$ such that 
\smallskip

\begin{itemize}
\item[(1)] $|Y|\le|X|$, 
\item[(2)] $Y$ dominates $G$,  
\item[(3)] $Y$ can be computed from $X$ in polynomial time, and 
\item[(4)] the number of connected components of $Y$ is at most the number of connected components of $X$.
\end{itemize} 
\end{proposition}

\begin{proof}
By Definition~\ref{def-dom-core-expansion}, if $X$ dominates $G$ then we are done. 
Otherwise, there exist $A \subseteq X$ and $B \subseteq V(G)$ such that $|B| < |A| \le c$, 
$X' = (X \setminus A) \cup B$ is a set that dominates $Z$, and the number of connected components of $X'$ is at most the number of connected components of $X$. 
Moreover, we can find such sets $A$ and $B$ by going through all possibilities in time $\Oh(|V(G)|^{2c-1})$
(i.e., in polynomial time for fixed $c$). By applying this argument iteratively on $X'$, 
we will eventually find a set $Y$ which dominates $G$. Furthermore, 
by Definition~\ref{def-dom-core-expansion}, $|B|< |A|$, and hence the size of the $Z$-dominator drops by one in each step. 
Therefore, the required set $Y$ can be found in time $\Oh(|V(G)|^{2c})$. 
\end{proof}

Clearly, $V(G)$ is a $c$-exchange domination core, for any $c$, but we look for a 
$c$-exchange domination core that is linear in $k$. Hence, we start with $Z = V(G)$ and
gradually reduce $|Z|$ by removing one vertex at a time, while
maintaining the invariant that $Z$ is a $c$-exchange domination core. 
To this end, we need to prove Lemma~\ref{lem:reduce-corce}. Note that we only remove 
vertices from $Z$ at this stage (no vertex deletions), and hence the graph remains intact. 

\begin{lemma}\label{lem:reduce-corce}
There exists a constant $C_{\textbf{core}} > 0$ depending only on a fixed (finite) number of grads of $\mathcal{G}$ and a 
polynomial-time algorithm that, given a graph $G \in \mathcal{G}$ and a 
$c$-exchange domination core $Z \subseteq V(G)$ with $|Z| > C_{\textbf{core}} \cdot k$, either correctly concludes that $\cds(G) > k$ or finds 
a vertex $z \in Z$ such that $Z \setminus \{z\}$ is still a $c$-exchange domination core.
\end{lemma}

The rest of the subsection is dedicated to proving Lemma~\ref{lem:reduce-corce}.
The algorithm of Lemma~\ref{lem:reduce-corce} consists of building a structural
decomposition of the graph $G$. More precisely, we identify a small set $X$ that dominates $G$, 
so that if $X$ was deleted from the graph, 
$Z$ would contain a large subset $S$, which is $2$-independent in the remaining graph.  
Given such a structure, we can argue that in any optimal $Z$-dominator, vertices of $X$ serve as dominators 
for almost all the vertices of $S$. This is because any vertex of $V(G)\setminus X$ can dominate at most one vertex from $S$. 
Since $S$ will be large compared to $X$, some vertices
of~$S$ will be indistinguishable from the point of view of domination via $X$, and these will be precisely the vertices that can be removed
from the domination core. The following lemma, which was proved by Drange et al.~\cite{DrangeDFKLPPRVS16} and builds on 
work by Dvorak~\cite{DBLP:journals/ejc/Dvorak13}, gives us such a decomposition. 

\pagebreak
\begin{lemma}\label{lem-decomp}
There exists a constant $C_{\textbf{Z}}' > 0$ depending only on a fixed (finite) number of grads of $\mathcal{G}$ and a polynomial-time 
algorithm that, given a graph $G \in \mathcal{G}$, an integer $k$, a constant $C_{\textbf{S}}' > 0$, and a set 
$Z' \subseteq V(G)$ with $|Z'| > C_{\textbf{Z}}' \cdot k$, either correctly concludes that $\ds(G) > k$ or finds 
a pair $(X',S')$ with the following properties:
\begin{itemize}
\item[(1)] $|X'| \leq C_{\textbf{X}}' \cdot k$, 
\item[(2)] $X'$ is a $Z'$-dominator in $G$, 
\item[(3)] for each $u \in V(G) \setminus X'$ we have $|M_3^G(u, X')| \leq C_{\textbf{M}}'$, and
\item[(4)] $S' \subseteq Z' \setminus X'$ is $2$-independent in $G - X'$, and $|S'| \geq C_{\textbf{S}}' \cdot |X|$, 
where $C_{\textbf{X}}'$ and $C_{\textbf{M}}'$ are constants depending only on a fixed number of grads of $\mathcal{G}$. 
\end{itemize}
\end{lemma}

\begin{corollary}\label{cor-decomp}
There exists a constant $C_{\textbf{Z}} > 0$ depending only on a fixed (finite) number of grads of $\mathcal{G}$ and 
a polynomial-time algorithm that, given a graph $G \in \mathcal{G}$, an integer~$k$, 
a constant $C_{\textbf{S}} > 0$, and a $c$-exchange domination core $Z \subseteq V(G)$ with $|Z| > C_{\textbf{Z}} \cdot k$, either correctly 
concludes that $\ds(G) > k$ or finds 
a pair $(X,S)$ with the following properties:

\begin{itemize}
\item[(1)] $|X| \leq C_{\textbf{X}} \cdot k$, 
\item[(2)] $X$ dominates $G$, 
\item[(3)] for each $u \in V(G) \setminus X$ we have $|M_3^G(u, X)| \leq C_{\textbf{M}}$, and
\item[(4)] $S \subseteq Z \setminus X$ is $2$-independent in $G - X$, and $|S| \geq C_{\textbf{S}} \cdot |X|$, 
where $C_{\textbf{X}}$ and $C_{\textbf{M}}$ are constants depending only on a fixed number of grads of $\mathcal{G}$. 
\end{itemize}
\end{corollary}
\begin{proof}
First we run the algorithm of Lemma~\ref{lem-decomp} for $G$, $k$, $Z' = Z$, $C_{\textbf{Z}}' = C_{\textbf{Z}}$, and constant  
$C_{\textbf{S}}' \ge C_{\textbf{S}}$, which we fix later in the proof. 
If the algorithm of Lemma~\ref{lem-decomp} concludes that $\ds(G) > k$, we correctly conclude that $\ds(G) > k$.
Otherwise, let ($X', S')$ be the output of the algorithm of Lemma~\ref{lem-decomp}. 
Since $Z$ is a $c$-exchange domination core, we find using Proposition~\ref{prop-dom-core}, a set $Y$, such 
that $|Y|\le |X'|$ and $Y$ dominates $G$. We let $X = \cl_3(Y\cup X')$ and $S = S'\setminus X$. Since $X$ is a superset of $Y$, $X$ satisfies 
property (2). Note that by Lemma~\ref{lem-closure}, $X$ can be computed in polynomial time and it satisfies property (3). 
Moreover,  $|X| = \cl_3(Y\cup X') \le C_{\textbf{cl1}} \cdot |Y \cup X'| \le 2 C_{\textbf{cl1}} \cdot |X'| \le 2 C_{\textbf{cl1}} \cdot C_{\textbf{X}}' \cdot k$ 
(Lemma~\ref{lem-closure}) and property (1) follows with $C_{\textbf{X}} = 2 C_{\textbf{cl1}} \cdot C_{\textbf{X}}'$. 
Hence $|S| \ge |S'|-|X| \ge C_{\textbf{S}}' \cdot |X'| - 2 C_{\textbf{cl1}} \cdot |X'| \ge 
(C_{\textbf{S}}' - 2 C_{\textbf{cl1}}) \cdot |X'| \ge \frac{(C_{\textbf{S}}' - 2 C_{\textbf{cl1}})}{2 C_{\textbf{cl1}}}|X|$. 
Therefore, if we set the constant 
$C_{\textbf{S}}' = 2 C_{\textbf{cl1}} (C_{\textbf{S}} + 2 C_{\textbf{cl1}})$, we satisfy all properties.
\end{proof}
Note that when $\ds(G) > k$ we can also conclude that $\cds(G) > k$. Hence, in the rest of the section we assume 
that we are given $G$, $Z$, and the constructed sets $X$ and $S$. We let $R = V(G) \setminus X$.
Using this notation, $S$ is $2$-independent in the graph $G[R]$. 
Define the following equivalence relation $\simeq$ on $S$: for $u,v\in S$, 
$
u\simeq v\quad \Leftrightarrow \quad \pr{i}{u}{X} = \pr{i}{v}{X}\textrm{ for each $1\leq i\leq 3$.}
$  Recall that for any vertex $u \in R$, we have $|\pr{3}{u}{X}|\leq C_{\textbf{M}}$. The following lemma follows directly 
from the proof of Lemma 3.8 by Drange et al.~\cite{DrangeDFKLPPRVS16}. 

\begin{lemma}\label{lem:num-classes1}
There exists a constant $C_{\textbf{eq}} > 0$ depending only on a fixed (finite) number of grads of $\mathcal{G}$ 
such that equivalence relation $\simeq$ has at most $C_{\textbf{eq}} \cdot 3^{C_{\textbf{eq}}} \cdot |X|$ classes.
\end{lemma}

\begin{proof}
From Lemma~\ref{lem-twin-expansion}, 
we know that the number 
of different $3$-projections in $X$ of vertices of $R$ is bounded by $C_{\textbf{ex}} \cdot |X|$. 
Observe that for each $u \in S$, we have $\pr{1}{u}{X}\subseteq \pr{2}{u}{X}\subseteq \pr{3}{u}{X}$ and the number of choices 
for $\pr{3}{u}{X}$ is again at most $C_{\textbf{ex}} \cdot |X|$ (since $S \subseteq R$). Moreover, since 
$u \in R$, we have that $|\pr{3}{u}{X}| \leq C_{\textbf{M}}$ (property (3) of Corollary~\ref{cor-decomp}). 
Hence, to define sets $\pr{i}{u}{X}$ for $1 \leq i < 3$ it suffices, for 
every $w \in \pr{3}{u}{X}$, to choose the smallest index $j$, $1 \leq j \leq 3$, such that $w\in \pr{j}{u}{X}$. The number of 
such choices is at most $3^{C_{\textbf{M}}}$, and hence the claim follows by setting $C_{\textbf{eq}} = C_{\textbf{ex}} + C_{\textbf{M}}$.
\end{proof}

We can now set the constant $C_{\textbf{S}}$ that is required in Corollary~\ref{cor-decomp} and the constant $c$ of the exchange domination core. 
We let $C_{\textbf{S}} = (4C_{\textbf{eq}} + 1) \cdot C_{\textbf{eq}} \cdot 3^{C_{\textbf{eq}}}$ and $c = (4C_{\textbf{eq}} + 1)$. 
Since we have that $|S|> C_{\textbf{S}} \cdot |X|$, from Lemma~\ref{lem:num-classes1} and the pigeonhole principle we 
infer that there is a class $\kappa$ of relation $\simeq$ with $|\kappa| > 4C_{\textbf{eq}} + 1$. We can find such 
class $\kappa$ in polynomial time by computing the classes of $\simeq$ directly from the definition and examining their sizes. 
We are ready to prove the final lemma of this section: any vertex of $\kappa$ can be removed from the 
$c$-exchange domination core $Z$ (recall that $S\subseteq Z$), which concludes the proof of Lemma~\ref{lem:reduce-corce}.

\begin{lemma}\label{lem:irrelevant}
Let $z \in \kappa$. Then $Z \setminus \{z\}$ is a $c$-exchange domination core, where $c = (4C_{\textbf{eq}} + 1)$.
\end{lemma}

\begin{proof}
Let $D \subseteq V(G)$ be a  set that dominates $Z' = Z \setminus \{z\}$ such that $D$ does not satisfy 
condition (2) of Definition~\ref{def-dom-core-expansion} with respect to $Z'$ (as otherwise we are done).  
In particular, this implies that $D$ cannot satisfy condition (2) with respect to $Z$ either. In other words, 
there exist no sets $A \subseteq D$ and $B \subseteq V(G)$ such that $|B| < |A| \le c$,  $(D \setminus A) \cup B$ has at most as 
many connected components as $D$,  and $(D \setminus A) \cup B$ dominates $Z$. 
Therefore, if $D$ dominates all of~$Z$, then $D$ has to satisfy condition (1) and dominate $G$. 
Hence, $z$ is not dominated by~$D$. We prove that this leads to a contradiction. 
 
Every vertex $s \in \kappa \setminus \{z\}$ is dominated by $D$ (since $\kappa \setminus \{z\} \subseteq Z'$). 
For each such $s$, let $v(s)$ be an arbitrarily chosen vertex of $D$ that dominates $s$. If $v(s) \in X$, then $v(s)$ also dominates~$z$, 
since $M_1(s,X) = M_1(z,X)$. Consequently, $v(s) \not\in X$ (because $D$ does not dominate $Z$) and the vertices $v(s)$ 
are pairwise different for all $s\in \kappa\setminus\{z\}$, as $S$ (and $\kappa$) is a $2$-independent set.
Let $W'= \{v(s) : s \in \kappa \setminus \{z\}\}$. Since $|\kappa| > 4C_{\textbf{eq}} + 1$, we have $|W'| \ge c = 4C_{\textbf{eq}} + 1$. 
Let $W$ be a set of arbitrarily chosen $c$ vertices of $W'$. 
Define $D' = (D \setminus W) \cup M_3(z,X)$. 

We first show that $D'$ dominates $Z'$. Towards that, it is enough to show that every vertex in $N[W]$ is dominated by $M_3(z,X) \subseteq D'$. 
By property (2) of Corollary~\ref{cor-decomp}, $X$ is a dominating set in $G$.
Hence, every vertex in $N[W] = N(W) \cup W$ either belongs to $X$ or has a neighbor in $X$. Let $Y \subseteq X$ be the set of vertices 
in $X$ dominating $N[W]$. Since every vertex in $N[W]$ is at distance at most two from some vertex in $\kappa$ 
(because each vertex in $W$ is adjacent to a vertex in $\kappa$), all vertices in $Y$ 
are at distance at most three from some vertex in $\kappa$. This implies that $Y \subseteq M_3(z,X)$ and therefore $D'$ dominates $Z'$. 

Note that $|M_3(z,X)| < |W| \leq c$ and therefore $D'$, which still dominates $Z'$, violates condition (2) of Definition~\ref{def-dom-core-expansion}. 
However, the set $D'$ can have more connected components than $D$ has. Therefore, we add additional vertices to $D'$ to ensure connectivity and still violate condition (2). 
For every vertex $u \in M_3(z,X)$ there is a path $P_u$ of length at most three between $u$ and $z$. 
Let $P = \bigcup_{u\in M_3(z,X)} P_u$ and consider the set $D'' = D'\cup P$. 
Note that $\vert P \vert \leq 3 \vert M_3(z,X)\vert \leq 3C_{\textbf{eq}}$. 

It remains to show that $D''$ has at most as many connected components as $D$.
If a component $C$ in $D$ contains a vertex from $W$, then each connected component $C'$ of $C\setminus W$ contains a 
vertex $v_{C'}\in N(W)$ and as mentioned before, $v_{C'}$ is adjacent to a vertex in $M_3(z,X)$. But in $D''$ all vertices of $M_3(z,X)$ are in one 
component due to the vertices of~$P$. Hence all such components of $D$ contribute to only one component of $D''$. On the other hand, if $C$ does not 
contain a vertex in $W$, then clearly all vertices of $C$ are also in $D''$ and hence $C$ is connected in $D''$ as well. 
%
It follows that  $D''$ has at most as many connected components as~$D$, $D''$ dominates $Z'$, and  
$|M_3(z,X)\cup P| \leq C_{\textbf{eq}} + 3C_{\textbf{eq}} < |W| = c = 4C_{\textbf{eq}} + 1$, contradicting the fact that $D$ did 
not satisfy condition (2) of Definition~\ref{def-dom-core-expansion}. 
\end{proof}

\subsection{Reducing connectors and dominators}
Armed with a $c$-exchange domination core $Z$ whose size is linear in $k$, our next goal is 
to reduce the number of connectors and dominators (the number of vertices in $V(G) \setminus Z$). 
To that end, we need the following lemma which is a generalized version of Lemma 2.11 in~\cite{DrangeDFKLPPRVS16}. 

\begin{lemma}[Trees closure lemma]\label{lem-tree-closure}
Let $\mathcal{G}$ be a class of bounded expansion and let $q$ and~$r$ be positive integers. 
Let $G \in \mathcal{G}$ be a graph and $X \subseteq V(G)$. Then a superset of vertices 
$X' \supseteq X$ can be computed in polynomial time, with the following properties: 
(1) For every $Y \subseteq X$ of size at most $q$, if $\stt_G(Y) \leq rq$ then $\stt_{G[X']}(Y) = \textbf{st}_G(Y)$. 
(2) $|X'| \leq C_{\textbf{tc}} \cdot |X|$, where $C_{\textbf{tc}}$ is a constant depending only on $r$, $q$, and a finite number of grads of $\mathcal{G}$. 
\end{lemma}

\begin{proof}
First, using Lemma~\ref{lem-closure} we compute $X_0 = \cl_{rq}(X)$. Then, $|X_0| \leq C_{\textbf{cl1}} \cdot |X|$ and for 
each vertex $u \notin X_0$ we have $|\prg{rq}{u}{X_0}{G}| \leq C_{\textbf{cl2}}$. 
Now, for each set $Y \subseteq X_0$ of at most $q$ vertices, compute an optimal Steiner tree $T_Y$ whose edges 
do not belong to $G[X_0]$; in case there is no such tree, set $T_Y = \emptyset$. 
Note that $T_Y$ can be computed in polynomial time for any fixed $q$~\cite{Bjorklund07}. 
Define $X'$ to be $X_0$ plus the vertex sets of all trees $T_Y$ that have size at most $rq$. 

\begin{claim}\label{cl:blowup2}
$|X'|\leq C_{\textbf{tc}} \cdot |X_0|$, where $C_{\textbf{tc}}$ is a constant depending only on $r$, $q$, and a finite number of grads of $\mathcal{G}$. 
\end{claim}

\begin{proof}[Proof of the Claim]
Let $H$ be a graph on vertex set $X_0$, where $uv \in E(H)$ if and only if there exists $Y$ such that $\{u,v\} \subseteq Y$, 
$T_Y \neq \emptyset$ and has size at most $rq$, and hence its vertex set was added to $X$. Note that we do not add multiedges.  
For every such set $Y$, $H[Y]$ induces a clique in $H$. Let $\omega(H)$ denote the number of cliques in $H$. 
Clearly $|X'| \leq |X_0| + rq \cdot \omega(H)$, so it suffices to prove an upper bound on $\omega(H)$. 

Consider an edge $uv \in E(H)$. The existence of this edge implies that $u$ and $v$ appear together in some tree $T_Y$ of size at most $rq$. 
Since $T_Y$ does not contain any edges from $G[X_0]$ (by construction), there must exist a path $P_{u,v}$ of length at most $rq$ 
connecting $u$ and~$v$. The internal vertices of $P_{u,v}$ do not belong to $X_0$. 
Take any $w \in X' \setminus X_0$, and consider for how many pairs $\{u,v\} \subseteq X_0$ it can hold that $w \in P_{u,v}$. If $\{u,v\}$ is 
such a pair, then in particular $u,v \in \prg{rq}{w}{X_0}{G}$. But we know that $|\prg{rq}{w}{X_0}{G}|\leq C_{\textbf{cl2}}$, so the 
number of such pairs is at most $\tau \leq (C_{\textbf{cl2}})^{2}$. Consequently, we observe that graph $H$ is 
an $(rq - 1)$-shallow minor of $G \odot K_\tau$: when each vertex $w \in X' \setminus X_0$ is 
replaced with $\tau$ copies, then we can realize all the paths between $u$ and $v$, in $G \odot K_\tau$, so that they are internally vertex-disjoint. 
From Lemma~\ref{lem-lex}, we know that $\grad_{rq-1}(G \odot K_\tau)$ is bounded polynomially in $\grad_{rq-1}(G)$ and $\tau$, which in turn is 
also bounded polynomially in $\grad_{rq-1}(\mc G)$. Hence $\grad_{rq-1}(G \odot K_\tau)$ is bounded polynomially 
in $\grad_{rq-1}(\mc G)$. The number of cliques in graph of bounded expansion is linear in the number of vertices~\cite{CLASSES}.
Combining the fact that $H$ has bounded expansion with $|X'| \leq |X_0| + rq \cdot \omega(H)$, the claim follows.
\end{proof}

\begin{claim}\label{cl:correctness2}
If $Y \subseteq X_0$ has size at most $q$ and $\stt_G(Y) \leq rq$, then $\stt_{G[X']}(Y) = \textbf{st}_G(Y)$.
\end{claim}

\begin{proof}[Proof of the Claim]
Let $T_Y$ be an optimal Steiner tree for $Y$ in $G$, and let $T_1, T_2, \ldots, T_p$ be the subtrees of size greater than one obtained after deleting 
all edges of $T_Y$ for which both endpoints are in $X_0$. Note that deleting such edges can only create either singleton vertices or subtrees 
of size greater than one. Moreover, let $Y_i$, $1 \leq i \leq p$, denote the set $Y \cap V(T_i)$. 
The existence of $T_i$ certifies that some tree of size at most $|T_i|$ was added when constructing $X'$ from $X_0$, and 
hence $\stt_{G[X']}(Y_i) \leq |T_i|$. Consequently, we infer that
$$\stt_{G[X']}(Y) \leq \sum_{i=1}^{p} \stt_{G[X']}(Y_i) + |Y \setminus \bigcup_{i = 1}^{p}{Y_i}| \leq 
\sum_{i=1}^{p}|T_i| + |Y \setminus \bigcup_{i = 1}^{p}{Y_i}| \leq |T_Y| = \stt_G(Y).$$
The opposite inequality $\stt_{G[X']}(Y) \geq \stt_G(Y)$ follows directly from the fact 
that $G[X']$ is an induced subgraph of $G$. 
\end{proof}
Claim~\ref{cl:blowup2} and the fact that $|X_0| \leq C_{\textbf{cl1}} |X|$ prove property (2). 
Claim~\ref{cl:correctness2} and the fact that $X \subseteq X_0$ prove property (1).
\end{proof}

Now let $\dot{Z}$ be a superset of $Z$, which we will fix later 
in Lemma~\ref{lem:expansionconn}, such that $|\dot{Z}|=\Oh(k)$. We compute $Z' = \cl_1(\dot{Z})$ using Lemma~\ref{lem-closure}; then we have that $|Z'| = \Oh(|\dot{Z}|) = \Oh(k)$. 
Partition $V(G) \setminus Z'$ into equivalence classes with respect to the following relation $\simeq$: 
For $u,v\in V(G)\setminus Z'$, set:
$
u\simeq v\quad \Leftrightarrow \quad \pr{1}{u}{Z'} = \pr{1}{v}{Z'}.
$
From Lemma~\ref{lem-closure} we know that for each $u \in V(G) \setminus Z'$, it holds 
that $|\pr{1}{u}{Z'}| = \Oh(1)$. Moreover, Lemma~\ref{lem-twin-expansion} implies that the number of possible 
different projections $\pr{1}{u}{Z'}$ for $u \in V(G)\setminus Z'$ is at most $\Oh(|Z'|)$. 
Hence, using the same reasoning as in the proof of Lemma~\ref{lem:num-classes1} we obtain the following.

\begin{lemma}\label{cl:num-classes2}
The equivalence relation $\simeq$ has at most $\Oh(|Z'|)$ classes.
\end{lemma}

Construct a graph $\ddot{G}$ as follows. Start with $G$ and, for each equivalence class $\kappa$ of relation~$\simeq$, 
add a new vertex $u_{\kappa}$ which is connected to all vertices in $\kappa$. Let $U = \bigcup_{\kappa \in \simeq} \{u_{\kappa}\}$. 
Our next step is to apply Lemma~\ref{lem-tree-closure} to graph $\ddot{G}$ with $X = Z' \cup U$, $r = 2$, and $q = 2t$ (we fix the value of $t$ later). 
Before we do so, we need to show that $\ddot{G}$ still has bounded grads. 

\begin{lemma}\label{lem:Gdotbndedexp}
There is a function $f \colon \mathbb{N} \rightarrow \mathbb{R}$ such that 
for all $r$ we have $\nabla_r(\ddot{G}) \leq f(r)$. 
\end{lemma}

\begin{proof}
We first construct a bipartite graph $H$ with bipartition $(A,B)$ as follows. We add a vertex $x_z \in A$ for each vertex $z \in Z'$ and we 
add one vertex $y_\kappa \in B$ for each equivalence class $\kappa$ of relation $\simeq$. We add an edge $x_zy_\kappa \in E(H)$ whenever 
$z$ is in $\pr{1}{w}{Z'}$, for some $w \in \kappa$. Finally, delete any isolated vertices in $H$ (applies if some equivalence class has no neighbors in $Z'$). 
It is not hard to see that $H$ is a subgraph of $G$ and, consequently, has bounded expansion. We can therefore apply Lemma~\ref{lem-charging} 
to obtain a mapping $\phi \colon B \rightarrow A$ with $y\phi(y) \in E(H)$, for each $y \in B$, and 
$\phi^{-1}(x_z) \leq C_{\textbf{ch}}$, for each $x_z \in A$. 

We now consider the graph $\dddot{G}$ which is obtained by adding a universal vertex to $G \odot K_{C_{\textbf{ch}}}$. 
From Lemma~\ref{lem-lex}, we know that $G \odot K_{C_{\textbf{ch}}}$ has bounded expansion. Applying Lemma~\ref{lem-uvertex} to $\dddot{G}$, we know that 
$\dddot{G}$ also has bounded expansion. Hence, it remains to show that $\ddot{G}$ is a subgraph of $\dddot{G}$. 
From the mapping $\phi$, we can associate each (except at most one) equivalence class with some vertex in $Z'$. 
In addition, every vertex in $Z'$ is associated with at most $C_{\textbf{ch}}$ classes. In other words, 
when each vertex in $Z'$ is replaced by $C_{\textbf{ch}}$ copies, each equivalence obtains a distinct universal vertex. 
The extra universal vertex added to $G \odot K_{C_{\textbf{ch}}}$ guarantees that the equivalence class with no neighbors in $Z'$ 
is also covered.  This completes the proof. 
\end{proof}

\begin{lemma}
\label{lem:expansionconn}
Let $\mathcal{G}$ be a class of bounded expansion and $(G,k)$ be an instance of {\sc Connected Dominating Set}, where $G\in \mathcal{G}$. 
Let $C_{\textbf{S}} = (4C_{\textbf{eq}} + 1) \cdot C_{\textbf{eq}} \cdot 3^{C_{\textbf{eq}}}$ and $c = (4C_{\textbf{eq}} + 1)$. 
Then, for any fixed $\epsilon > 0$, there is a polynomial-time algorithm that  
either concludes that $\cds(G)>k$ or outputs a set $Y\subseteq V(G)$ of cardinality 
$\Oh(f(\epsilon) \cdot k)$, for some function~$f$, and a set $Z \subseteq Y$, such that $(i)$ $Z$ is a $c$-exchange domination core in $G$ and 
$(ii)$ $\text{OPT}_{\text{SCDS}}((G[Y],Z),k)\leq (1+\epsilon)\text{OPT}_{\text{CDS}}(G,k)$.  
\end{lemma}

\begin{proof}
We start by designing an algorithm ${\cal A}$ with the desired properties.  
Starting with Lemma~\ref{lem:reduce-corce},
Algorithm ${\cal A}$ either 
concludes that $\cds(G) > k$ or finds (in polynomial time) a $c$-exchange domination core $Z$ of size at most $\Oh(k)$. 
Since $Z$ is a $Z$-dominator itself, using Proposition~\ref{prop-dom-core}, we find a set $O$, such that $|O|\le |Z|$ and $O$ dominates $G$. 
Moreover, by Proposition~\ref{approx-cds-by-ds}, we find a connected superset $\ddot{O}$ of $O$, such that $|\ddot{O}|\le 3|O|$.
Next, we let $\dot{Z}=Z\cup \ddot{O}$ and compute $Z' = \cl_1(\dot{Z})$ using Lemma~\ref{lem-closure}; then we have that $|Z'| = \Oh(k)$. 
We partition $V(G) \setminus Z'$ into equivalence classes with respect to $\simeq$ (i.e. $u\simeq v\quad \Leftrightarrow \quad \pr{1}{u}{Z'} = \pr{1}{v}{Z'}$). 
From Lemma~\ref{cl:num-classes2}, we know that the equivalence relation $\simeq$ has at most $\Oh(k)$ classes. 
Let $X'$ be the set of size $\Oh(k)$ obtained by applying Lemma~\ref{lem-tree-closure} to graph 
$\ddot{G}$ with $X = Z' \cup U$, $r = 2$, and $q = 2t$ (we fix the value of $t$ later and $U$ is the set defined earlier). 

Algorithm ${\cal A}$ outputs instance $((G',Z),k)=((G[Y],Z),k)$ which asks to dominate $Z$, 
where  $Y= Z'\cup  (X' \setminus U)$.  
Note that, since $G[Y]$ contains a connected dominating set $\ddot{O}$ of~$G$, it is connected as well.
Moreover, the total number of vertices 
in $G[Y]$ is at most $\Oh(k)$. 

We now show that $\text{OPT}_{\text{SCDS}}((G[Y],Z),k)\leq (1+\epsilon)\text{OPT}_{\text{CDS}}(G,k)$. 
Consider the graph $D^*$ induced by an optimal connected dominating set of $G$. 
If  $\vert V(D^*) \vert > k$, then $\text{OPT}_{\text{SCDS}}((G[Y],Z),k)\leq (1+\epsilon)\text{OPT}_{\text{CDS}}(G,k)$ holds 
trivially. So we assume that $\vert V(D^*) \vert \leq k$. 
We let $\mathcal{F}(D^*,t) = \{T_1, T_2, \cdots,T_\ell\}$ denote a $(D^*,t)$-covering family. 
Proposition~\ref{prop-covering-family} implies that 
$\ell \leq \frac{|V(D^*)|}{t} + 1$ 
and $\sum_{T \in \mathcal{F}(D^*,t)}{|V(T)|} \leq (1 + \frac{1}{t}) |V(D^*)| + 1$. 
Moreover, the size of each subtree $T$ is at most $2t$. We define groups on $V(G)$ as follows. 
Each vertex $v \in Z'$ belongs to its own unique group $q_v$ and each vertex in $V(G) \setminus Z'$ belongs to group $q_\kappa$, i.e. 
vertices in the same equivalence class in $\simeq$ belong to the same group. 

\begin{claim} \label{claim:expansion1}
For any $T \in \mathcal{F}(D^*,t)$ there exists a tree $T'$ in $G'$ of size at most $|V(T)|$ which contains 
at least one vertex from each group appearing in $T$. 
\end{claim}

\begin{proof}
Recall that, when constructing the graph $\ddot{G}$, we added one universal vertex for each equivalence class in $V(G) \setminus Z'$. 
Hence, after applying Lemma~\ref{lem-tree-closure}, we know that for any $Y \subseteq Z' \cup U$ of size at 
most $2t$ (which is exactly a subset of the groups) if $\stt_G(Y) \leq 4t$ then $\stt_{\ddot{G}}(Y) = \textbf{st}_G(Y)$. Every vertex in $Z'$ belongs to 
a distinct group and every vertex $u_{\kappa} \in U$ is connected to all vertices in $\kappa$. 
Hence, any tree of size greater than one containing a vertex $u_{\kappa}$ must also contain a neighbor of $u_{\kappa}$ (from the same group). 
The existence of $T$ implies that there exists a tree of size at most $2t$ connecting all groups appearing in $T$. 
Hence, a tree $T'$ certifying this fact exists in~$G'$.   
\end{proof}

We now construct a new family $\mathcal{F}'$ which consists of replacing each $T \in \mathcal{F}(D^*,t)$ 
by a set $T'$ in $G'$.  Let $D'' = \bigcup_{T' \in \mathcal{F}'}V(T')$. By the previous claim and the fact 
that $\sum_{T \in \mathcal{F}(D^*,t)}{|V(T)|} \leq (1 + \frac{1}{t}) |V(D^*)| + 1$, we know that $|D''| \leq (1 + \frac{1}{t}) |V(D^*)| + 1$ 
and $D''$ dominates $Z'$ (since we never reduce the number of groups in $D''$). 
Moreover, $D''$ consists of at most $\frac{|V(D^*)|}{t} + 1$ components. Note that $\ddot{O}$ is a connected $Z$-dominator and therefore the set $\ddot{O}\cup Z$ is also connected. 
Now, since $Z\cup \ddot{O}\subseteq{Z'}$, $D''$ dominates $Z\cup \ddot{O}$ and applying Proposition~\ref{approx-cds-by-ds}, we obtain 
a connected $(Z\cup \ddot{O})$-dominator, and hence also $Z$-dominator, $D'$ of size at most 
$\frac{2|V(D^*)|}{t} + 2 + (1 + \frac{1}{t}) |V(D^*)| + 1 = \frac{2|V(D^*)|}{t} + 3 + |V(D^*)| + \frac{|V(D^*)|}{t} = (1 + \frac{3}{t})|V(D^*)| + 3$.  
Setting $t = \frac{6}{\epsilon}$ implies that $D'$ is a $(1 + \epsilon)$-approximate solution, for $\text{OPT}_{\text{CDS}}(G,k) \geq \frac{6}{\epsilon}$ (note that 
$3 \leq \frac{\epsilon}{2} \cdot \text{OPT}_{\text{CDS}}(G,k)$). 
We can assume $\text{OPT}_{\text{CDS}}(G,k) \geq \frac{6}{\epsilon}$ without loss of generality, as otherwise a simple brute-force algorithm runs in polynomial time. 
\end{proof}

\pagebreak
Using the same arguments as in the proof of Theorem~\ref{thkddfree} and replacing Lemma~\ref{cdcore} 
by Lemma~\ref{lem:expansionconn}, we obtain the following theorem. Note however that the solution-lifting algorithm 
might have to apply Proposition~\ref{prop-dom-core} to guarantee domination of the input graph.  

\begin{theorem}
For every $\epsilon > 0$, CDS 
admits a $(1 + \epsilon)$-approximate bi-kernel with $\Oh(f(\epsilon) \cdot k)$ vertices on graphs of bounded expansion, 
where $f$ is some computable function.  
\end{theorem}

\section{Nowhere dense graphs}\label{sec:nd}
In the following, we fix a nowhere dense class $\Cc$ of 
graphs, $k,r\in \N$ and $\alpha>1$. Furthermore, let
$t= \frac{\alpha-1}{4r+2}$. As we deal with the {\sc Connected Distance-$r$ Dominating Set} problem we may assume
that all graphs in $\Cc$ are connected. 

\begin{definition}
Let $G$ be a graph. A set $Z\subseteq V(G)$ is a \emph{$(k,r)$-domination core} for $G$ if every set $D$ of size at most $k$ that $r$-dominates $Z$ also $r$-dominates $G$ 
\end{definition}

Domination cores of polynomial size exist for nowhere dense
classes, as the following lemma shows. 

\begin{lemma}[Kreutzer et al.~\cite{siebertz2016polynomial}]
\label{lem:findcore1}
There exists a polynomial $q$ (of degree depending only on~$r$) and a polynomial time algorithm that, given a graph $G\in\Cc$ and $k\in\N$
either correctly concludes that $G$ cannot be $r$-dominated by a 
set of at most $k$ vertices, or finds a $(k,r)$-domination core $Z\subseteq V(G)$ of $G$ of size at most $q(k)$. 
\end{lemma}

We remark that the non-constructive
polynomial bounds that follow from~\cite{siebertz2016polynomial}
can be replaced by much improved constructive bounds~\cite{pilipczuk2017wideness}. 

\medskip
We will work with the following parameterized
minimization variant of the {\sc Connected Distance-$r$ Dominating Set} problem. 

\[\textsc{$r$-CDS}(G,k,D)=
\begin{cases}
\infty & \text{if $D$ is not a connected distance-$r$}\\
&\quad \text{dominating set of $G$}\\\min\{|D|,k+1\} & \text{otherwise.}
\end{cases}\]

As indicated earlier, we compute only a bi-kernel and reduce
to the following annotated version of $r$-CDS. 

\[\textsc{$r$-SCDS}((G,Z),k,D)=
\begin{cases}
\infty & \text{if $D$ is not a connected distance-$r$}\\
&\quad \text{dominating set of $Z$ in $G$}\\\min\{|D|,k+1\} & \text{otherwise.}
\end{cases}\]

\medskip

The following lemma is proved just as Proposition~\ref{approx-cds-by-ds}.

\begin{lemma}\label{lem:ds-cds}
Let $G$ be a graph, $Z\subseteq V(G)$ a connected set in $G$ and 
let $D$ be a distance-$r$ dominating set for $Z$ such that $G[D]$ has at most 
$p$ connected components. Then we can compute in polynomial time 
a set $Q$ of size at most
$2rp$ such that $G[D \cup Q]$ is connected.
\end{lemma}

\medskip
The lemma implies that we may assume that our
domination cores are connected. 

\begin{corollary}\label{lem:findcore}
There exists a polynomial $q$ (of degree depending
only on $r$) and a polynomial time algorithm
that, given a graph $G\in\Cc$ and $k\in\N$
either correctly concludes that~$G$ cannot be $r$-dominated by a 
set of at most $k$ vertices, or finds a 
$(k,r)$-domination core $Z\subseteq V(G)$ of~$G$ of size at most $q(k)$ such that $G[Z]$ is connected. 
\end{corollary}

\begin{proof}
Assume that when applying Lemma~\ref{lem:findcore1}, a 
$(k,r)$-domination core $Y$ is returned, otherwise we
return that no distance-$r$ dominating set of size at most
$k$ exists. 

First observe that every superset $X\supseteq Y$
is also a $(k,r)$-domination core of $G$ (every set of size at most $k$ which $r$-dominates $X$ in particular $r$-dominates $Y$, and
hence all of $G$). 

Assume there is a vertex $v\in V(G)$ with distance greater 
than $2r$ from $Y$. Since $Y$ is a $(k,r)$-domination core, 
every set of size at most $k$ that $r$-dominates $Y$ also $r$-dominates~$G$. If there exists a distance-$r$ dominator $A$ of $Y$ 
of size at most $k$, also $B=N_r[Y]\cap A$ 
(the intersection of~$A$ with the
closed $r$-neighborhood of $Y$) is a distance-$r$ dominator of 
$Y$ of size at most $k$. However, as $v$ has distance 
greater than $2r$ from $Y$, $B$ cannot be a distance-$r$
dominating set of $G$. Hence, if there is $v\in V(G)$ with
distance greater than $2r$ from $Y,$ we may return 
that~$G$ cannot be $r$-dominated by a set of at most $k$
vertices. Otherwise, it follows that $Y$ is a distance-$2r$ 
dominating set of $G$. We can hence apply
Lemma~\ref{lem:ds-cds} with parameters $Z=V(G)$ (we assume
that all graphs $G\in\Cc$ are connected) and $D=Y$
to find a connected set $X\supseteq Y$ of size at most
$(2r+1)\cdot q(k)$ which is a connected $(k,r)$-domination core. 
\end{proof}

We will need a precise description of how vertices interact with the domination core.


\begin{definition}
The {\em{$r$-projection profile}} of a vertex $u\in V(G)\setminus A$ on $A$ is a function $\rho^G_r[u,A]$ mapping vertices of
$A$ to $\{0,1,\ldots,r,\infty\}$, defined as follows: for every $v\in A$, the value $\rho^G_r[u,A](v)$ is the length of a shortest $A$-avoiding path (a path whose internal
vertices do not belong to $A$) connecting $u$ and~$v$, and~$\infty$ in case this length
is larger than $r$. We define $\projprof_r(G,A)=|\{\rho_r^G[u,A]\colon u\in V(G)\setminus A\}|$
to be the number of different $r$-projection profiles realized on $A$. 
\end{definition}

\begin{lemma}[Eickmeyer et al.~\cite{eickmeyer2016neighborhood}]\label{lem:projection-complexity}
There is a function $\fproj$ such that for every
$G\in \Cc$, vertex subset $A\subseteq V(G)$, and
  $\epsilon>0$ 
  we have $\projprof_r(G,A)\leq \fproj(r,\epsilon)\cdot |A|^{1+\epsilon}$.
\end{lemma}

The following lemma is immediate from the definitions. 
\begin{lemma}\label{lem:dswithproj}
Let $G$ be a graph and let $X\subseteq V(G)$. Let $D$ be a distance-$r$ dominating set of~$X$. Then every set $D'$ such
that for each $u\in D$ there is $v\in D'$ with 
$\rho_r^G[u,X]=\rho_r^G[v,X]$ is a distance-$r$ dominating 
set of $X$. 
\end{lemma}

The following generalization of the \emph{Tree Closure Lemma},
Lemma~\ref{lem-tree-closure},
shows that we 
can  re-combine the pieces in
nowhere dense graph classes. 

\begin{lemma}\label{lem:tree-closure}
There exists
a function $f$ such that the
following holds. Let $G\in\Cc$, let $X\subseteq V(G)$, 
and let
$\epsilon>0$. Define an equivalence relation 
$\sim_{X,r}$ on $V(G)$ by 
\[u\sim_{X,r}v\Leftrightarrow \rho_r^G[u,X]=\rho_r^G[v,X].\]
Then we can compute in time $\Oof(|X|^{t(1+\epsilon)}\cdot n^{1+\epsilon})$ a subgraph $G'\subseteq G$ of $G$ 
such that 
\begin{align*}
\textit{1)} & \quad X\subseteq V(G'), \\
\textit{2)} & \quad \text{for every $u\in V(G)$ there 
is $v\in V(G')$ with $\rho_r^G[u,X]=\rho_r^{G'}[v,X]$},\\
\textit{3)} & \quad \text{for every set~$\YYY$ of at most $2t$ projection classes (equivalence classes of $\sim_{X,r}$),}\\
& \quad \quad\text{if $\mathbf{st}_G(\YYY)\leq 2t$, then $\mathbf{st}_{G'}(\YYY)=\mathbf{st}_G(\YYY)$, and }\\
\textit{4)} & \quad |V(G')|\leq f(r,t,\epsilon)\cdot |X|^{2+\epsilon}.
\end{align*}
Note that in item \textit{3)}, due to item \textit{2)}, 
every class of $\sim_{X,r}$ which is non-empty
in $G$ is also a non-empty class of $\sim_{X,r}$ in $G'$. 
\end{lemma}

We defer the proof of the lemma to the next subsection. 

\begin{lemma}\label{lemma:pre-kernel}
Let $\epsilon>0$ and let $q$ be the 
polynomial from Lemma~\ref{lem:findcore}. There exists an algorithm running in time $\Oof(q(k)^{t(1+\epsilon)}\cdot n^{1+\epsilon})$ that, given an $n$-vertex graph 
  $G\in \Cc$ and a positive integer $k$, either returns that there exists
  no connected distance-$r$ dominating set of $G$, or 
  returns a subgraph $G'\subseteq G$ and a vertex subset $Z\subseteq V(G')$ with the following properties:
  \begin{align*}
  \textit{1)} &\quad \text{$Z$ is a $(k,r)$-domination 
  core of $G$,}\\
  \textit{2)}&\quad \text{$\Opt_{\textsc{$r$-SCDS}}((G',Z),k)\leq 
  \alpha\cdot \Opt_{\textsc{$r$-CDS}}(G,k)$, and}\\
  \textit{3)}&\quad \text{$|V(G')|\leq p(k)$, for some polynomial $p$ whose degree depends only on $r$.}
  \end{align*}
\end{lemma}
\begin{proof}
Using Lemma~\ref{lem:findcore}, we first conclude that $G$ cannot be $r$-dominated by a 
connected set of at most $k$ vertices, or we find a 
connected $(k,r)$-domination core $Z\subseteq V(G)$ of $G$ of size at most 
$q(k)$.
In the first case, we reject the instance, otherwise, 
let $G'\subseteq G$ be the subgraph that we obtain 
by applying Lemma~\ref{lem:tree-closure} with parameters
$G,Z,t$ and $\epsilon$. Let $p;=
f(r,t,\epsilon)\cdot q^{2+\epsilon}$ (where~$f$
is the function from Lemma~\ref{lem:tree-closure}), which is a polynomial
of degree depending only on~$r$, 
only the coefficients depend on $\alpha$. 

It remains to show that $\Opt_{\textsc{$r$-SCDS}}((G',Z),k)\leq \alpha\cdot \Opt_{\textsc{$r$-CDS}}(G,k)$. Let $D^*$ be a 
minimum connected distance-$r$ dominating set of $G$ of size at most $k$ (if $|D^*|>k$, then 
$\Opt_{\textsc{$r$-SCDS}}((G',Z),k)\leq \alpha\cdot \Opt_{\textsc{$r$-CDS}}(G,k)$ trivially holds). 
Let $\mathcal{F}=\mathcal{F}(G[D^*],t)=\{T_1,\ldots, T_\ell\}$ 
be a covering family for the connected graph $G[D^*]$ 
obtained by Proposition~\ref{prop-covering-family}. Note that by the lemma we 
have $\ell\leq |D^*|/t+1$ and $\sum_{1\leq i\leq \ell}V(T_i)\leq (1+1/t)|D^*|+1$. 
Moreover, the size of each subtree $T_i$ is at most $2t$. 
By construction of $G'$ (according to item~\textit{3)} of 
Lemma~\ref{lem:tree-closure}), 
for each $T\in \mathcal{F}$ there exists a tree $T'$ in $G'$ 
of size at most $|V(T)|$ which contains for each $u\in V(T)$
a vertex $v$ with $\rho_r^G[u,Z]=\rho_r^{G'}[v,Z]$.

We construct a new family $\mathcal{F}'$ which we obtain by replacing each $T\in \mathcal{F}$ by the tree~$T'$ described
above. Let $D':=\bigcup_{T'\in \mathcal{F}'}V(T')$ in $G'$. 
We have $\sum_{T'\in \mathcal{F}'}|V(T')|\leq(1+1/t)|D^*|+1$ and
since $D'$ contains
vertices from the same projection classes as $D^*$, according to Lemma~\ref{lem:dswithproj}, $D'$ is a distance-$r$ dominating set of $Z$. Moreover, $G[D']$ has at 
most $\ell\leq |D^*|/t+1$ components. We apply Lemma~\ref{lem:ds-cds}, and obtain
a set $Q$ of size at most $2r(|D^*|/t+1)$ such that 
$D''=D'\cup Q$ is
a connected distance-$r$ dominating set of $Z$. 
We hence have \[|D''|\leq 2r(|D^*|/t+1)+(1+1/t)|D^*|+1=(1+\frac{2r+1}{t})|D^*|+2r+1\leq (1+\frac{4r+2}{t})|D^*|\] (we may assume
that $2r+1\leq \frac{2r+1}{t}|D^*|$, as otherwise we can simply run a brute force algorithm in polynomial time). We conclude by 
recalling that $t= \frac{\alpha-1}{4r+2}$.
\end{proof}

\begin{theorem}\label{thm:lossykernel}
There exists a polynomial $p$ whose
degree depends only on $r$ such that 
the {\sc Connected Distance-$r$ Dominating Set} problem on $\Cc$ 
admits an $\alpha$-approximate bi-kernel with 
$p(k)$ vertices. 
\end{theorem}

\begin{proof}
The $\alpha$-approximate polynomial-time pre-processing
algorithm first calls the algorithm of
Lemma~\ref{lemma:pre-kernel}. If 
it returns that there exists no distance-$r$
dominating set of size at most~$k$ for~$G$, we return 
a trivial negative instance. Otherwise, 
let $((G',Z),k)$ be the annotated instance returned by 
the algorithm. 
The solution lifting algorithm, 
given a connected distance-$r$ dominating set of $Z$
in $G'$ simply returns $D$.

By construction of $G'$ we have $M_r^{G'}(u,Z)\subseteq M_r^G(u,Z)$ 
for all $u\in V(G')$. Hence every distance-$r$ $Z$-dominator
in $G'$ is also a distance-$r$ $Z$-dominator in $G$. 
In particular, since $Z$ is a 
$(k,r)$-domination core, $D$ is also a connected 
distance-$r$ dominating set for $G$. 

Finally, by Lemma~\ref{lemma:pre-kernel}
we have $\Opt_{\textsc{$r$-SCDS}}((G',Z),k)\leq 
\alpha\cdot \Opt_{\textsc{$r$-CDS}}(G,k)$, which implies
\begin{align*}
\frac{{\textsc{$r$-CDS}}(G,k,D)}{\Opt_{\textsc{$r$-CDS}}(G,k)}\leq \alpha(k)\cdot \frac{{\textsc{$r$-SCDS}}((G',Z),k,D)}{\Opt_{\textsc{$r$-SCDS}}((G',Z),k)}. 
\end{align*}
\end{proof}

Observe that we obtain a $1$-approximate bi-kernel for
the distance-$r$ dominating set problem by just taking one 
vertex from each projection class of the $(k,r)$-domination
core. 

We remark that we obtain a bi-kernel and not a kernel for the
following reason. Formally, we are dealing with the connected
distance-$r$ dominating set problem on a class $\Cc$ of graphs. 
We reduce the problem to the problem of finding a connected
distance-$r$ dominating set of a marked set $Z$ of vertices
on the class $\Cc$. We could slightly change the reduction, 
see~\cite{DrangeDFKLPPRVS16}, so 
that we ask to dominate the whole graph, hence reduce to 
the connected distance-$r$ dominating set again, however, the
reduction would introduce a gadget which takes us out of the 
class $\Cc$. Hence, we reduce in a very strict sense to a 
different problem, namely, connected distance-$r$ domination on 
a class $\mathcal D$.

\subsection{The proof of Lemma~\ref{lem:tree-closure}}\label{sec:tree-closure}

Lemma~\ref{lem:tree-closure} is the most technical contribution
of the nowhere dense part. This whole section is devoted to its proof. 
We will mainly make use of a characterization of nowhere 
dense graph classes by the so-called \emph{weak 
coloring numbers}. 

\begin{definition}
For a graph $G$, by $\Pi(G)$ we denote the set of all linear orders
of $V(G)$.  For $u,v\in V(G)$ and
any $s\in\N$, we say that~$u$ is \emph{weakly $s$-reachable} from~$v$
with respect to a linear order $L$, if there is a path $P$ of length at most $s$ connecting $u$ and $v$ such that $u$ is 
the smallest among
the vertices of $P$ with respect to~$L$. By $\WReach_s[G,L,v]$ we
denote the set of vertices that are weakly $s$-reachable from~$v$ with
respect to~$L$. For any subset $A\subseteq V(G)$, we let
$\WReach_s[G,L,A] = \bigcup_{v\in A} \WReach_s[G,L,v]$.  The
\emph{weak $s$-coloring number $\wcol_s(G)$} of $G$ is defined as 
\begin{eqnarray*}
  \wcol_s(G)& = & \min_{L\in\Pi(G)}\:\max_{v\in V(G)}\:
                   \bigl|\WReach_s[G,L,v]\bigr|.
\end{eqnarray*}
\end{definition}

The weak coloring numbers were introduced by Kierstead and
Yang~\cite{kierstead2003orders} in the context of coloring and
marking games on graphs. As proved by Zhu \cite{zhu2009coloring},
they can be used to characterize both bounded expansion and nowhere
dense classes of graphs. In particular, we use the following.

\begin{theorem}[Zhu \cite{zhu2009coloring}]\label{lem:wcolbound}
  Let $\Cc$ be a nowhere dense class of graphs.
  There is a function $f_{\wcol}$ such that
  for
  all $s\in\N$, $\epsilon>0$, and $H\subseteq G\in \Cc$ we have
  $\wcol_s(H)\leq f_{\wcol}(s,\epsilon) \cdot |V(H)|^\epsilon$.
\end{theorem}

One can define artificial classes where the functions $f_\wcol$ grow 
arbitrarily fast, however, on many familiar sparse graph classes they are
quite same, e.g.\ on bounded tree-width graphs~\cite{GroheKRSS15}, 
graphs with excluded minors~\cite{siebertz16} or excluded topological
minors~\cite{KreutzerPRS16}. Observe that in any case
the theorem allows to pull polynomial blow-ups on the graph 
size to the function $\fwcol$. More precisely, for any $\epsilon>0$, 
if we deal with a subgraph of size $n^x$ for some $x\in \R$, 
by re-scaling $\epsilon$ to $\epsilon'=\epsilon/x$, we will 
get a bound of $\fwcol(s,\epsilon')\cdot (n^x)^{\epsilon'}
=\fwcol(s,\epsilon')\cdot n^\epsilon$ for the weak $s$-coloring 
number. 

Our second application of the weak coloring numbers is described in the next lemma, which shows that 
they capture the local separation properties of a 
graph. 

\begin{lemma}[see Reidl et al.~\cite{reidl2016characterising}]\label{lem:wcol-sep}
Let $G$ be a graph and let $L\in \Pi(G)$. Let $X\subseteq V(G)$, $y\in V(G)$
and let $P$ be a path of length at most $r$ between a vertex $x\in X$ and $y$. 
Then \[\big(\WReach_r[G,L,X]\cap \WReach_r[G,L,y]\big)\cap V(P)\neq \emptyset.\]
\end{lemma}
\begin{proof}
Let $z$ be the minimal vertex of $P$ with respect to $L$. Then both $z\in \WReach_r[G,L,x]$ and $z\in \WReach_r[G,L,y]$. 
\end{proof}

We are now ready to define the graph $G'$ whose
existence we claimed in the previous section. 

\begin{definition}\label{def:GX}
Let $G\in\Cc$ and fix a subset $X\subseteq V(G)$. 
Define an equivalence relation~$\sim_{X,r}$ on $V(G)$ by 
\[u\sim_{X,r}v\Leftrightarrow \rho_r^G[u,X]=\rho_r^G[v,X].\] 
For each subset~$\YYY$ of projection classes 
of size at most $2t$, if $\mathbf{st}_G(\YYY)\leq 2t$, 
fix a Steiner tree~$T_\YYY$ for~$\YYY$ of minimum size.
For such a tree $T_\YYY$ call a vertex $u\in \kappa\cap V(T_\YYY)$ with $\kappa\in \YYY$ a \emph{terminal} of~$T_\YYY$. 
We let $C=\{u\in V(G) : u$ is a terminal of some
$T_\YYY\}$. 

Let~$G'$ be a subgraph of $G$ which contains 
$X$, all $T_\YYY$ as above, and a set of vertices and 
edges such that $\rho_r^G[u,X]=\rho_r^{G'}[u,X]$
for all $u\in C$. 
\end{definition}

\begin{lemma}\label{lem:computeG_X}
There exist functions $f$ and 
$g$ such that 
for every $G\in\Cc$, $X\subseteq V(G)$ and $\epsilon>0$ 
we can compute 
a graph $G'$ as described above of size $f(r,\epsilon)\cdot
|X|^{2t(1+\epsilon)}$ in time $g(r,t,\epsilon)\cdot
|X|^{2t(1+\epsilon)}$. 
\end{lemma}
\begin{proof}
According to Lemma~\ref{lem:projection-complexity} there is a function
$\fproj$ such that for every $G\in \Cc$, vertex subset 
$A\subseteq V(G)$, and $\epsilon>0$ we have $\projprof_r(G,A)\leq \fproj(r,\epsilon)\cdot |A|^{1+\epsilon}$. We now apply the lemma
to $A=X$. 

We compute for each $v\in X$ 
the first $r$ levels of a breadth-first search 
(which terminates whenever 
another vertex of $X$ is encountered, as to compute $X$-avoiding
paths). For each visited vertex $w\in V(G)$ we remember the
distance to $v$. In this manner, we compute in time
$\Oof(|X|\cdot n^{1+\epsilon})$ the projection profile
of every vertex $w\in V(G)$. Observe that Lemma~\ref{lem:wcolbound}
applied to $r=1$ implies that an 
$n$-vertex graph $G\in \Cc$ is $n^\epsilon$-degenerate 
and in particular has only $\Oof(n^{1+\epsilon})$ many edges. 
Hence a breadth-first search can be computed in time
$\Oof(n^{1+\epsilon})$. 

We now decide for each subset $\YYY$ of at most $2t$ 
projection classes whether $\mathbf{st}_G(\YYY)\leq 2t$. 
If this is the case, we also compute a 
Steiner tree $T_\YYY$ of minimum size in time 
$h(t,\epsilon)\cdot n^{1+\epsilon}$ for some
function $h$. To see that this
is possible, observe that the problem is equivalent to testing 
whether an existential 
first-order sentence holds in a colored graph, which is possible
in the desired time on nowhere dense classes~\cite{grohe2014deciding, sparsity}.

Finally, for each sub-tree $T_\YYY$ and each $\kappa\in 
\YYY$ fix some terminal $u\in \kappa\cap V(T_\YYY)$. Compute the 
first~$r$ levels of an $X$-avoiding breadth-first search with
root $u$ and add the vertices and edges of the bfs-tree 
to ensure
that $\rho_r^G[u,X]=\rho_r^{G'}[u,X]$. Observe that
by adding these vertices 
we add at most $|X|\cdot r$ vertices for each vertex $u$. 

As we have  $\Oof\left(\left(|X|^{(1+\epsilon)}\right)^{2t}\right)=\Oof\left(|X|^{2t(1+\epsilon)}\right)$ many subsets
of projection classes of size at most $2t$, we can conclude by defining $f$ and $g$ appropriately. 
\end{proof}

It remains to argue that the graph $G'$ is in fact much smaller than 
our initial estimation in Lemma~\ref{lem:computeG_X}. First, as 
outlined earlier, we do not care about polynomial blow-ups when 
bounding the weak coloring numbers. 

\begin{lemma}\label{lem:wcolGX}
There is a function $h$ such that
for all $s\in \N$ and $\epsilon>0$
we have \[\wcol_{s}(G')\leq h(r,s,t,\epsilon)\cdot |X|^\epsilon.\]
\end{lemma}
\begin{proof}
Choose $\epsilon':= \epsilon/(3t)$. According to  Lemma~\ref{lem:computeG_X}, $G'$ has size at most 
$f(r,1/2)\cdot |X|^{3t}$ (apply the lemma with $\epsilon=1/2$). 
According to Lemma~\ref{lem:wcolbound}, we have 
\[\wcol_{s}(G')\leq \fwcol(s,\epsilon')\cdot \left(f(r,1/2)\cdot |X|^{3t}\right)^{\epsilon'}.\] Conclude by defining $h(r,s,t,\epsilon)=
\fwcol(s,\epsilon')\cdot f(r,1/2)^{\epsilon'}$. 
\end{proof}

Our next aim is to decompose the group Steiner trees
into single paths which are then analyzed with the help of
the weak coloring numbers. We need a few more
auxiliary lemmas. 


\medskip
The following two lemmas are easy consequences of 
the definitions. 

\begin{lemma}\label{lem:wcollex}
Let $G,H$ be graphs and let $s\in \N$. Then 
$\wcol_{s}(G\odot H)\leq |V(H)|\cdot \wcol_{s}(G)$. 
\end{lemma}


\begin{lemma}\label{lem:wcolsubdiv}
Let $G$ be a graph and let $r',s\in\N$. 
Let $H$ be any graph obtained by replacing some
edges of $G$ by paths of length $r$. 
Then $\wcol_{r'}(H)\leq s+\wcol_{r'}(H)$. 
\end{lemma}

To estimate the size of $G'$ we reduce the 
group Steiner tree problems to simple Steiner tree problems
in a super-graph $\dot{G}$ of $G'$. 

\begin{definition}
See Figure~\ref{fig:dotG} for an illustration of the following 
construction. Let $G'$ with distinguished terminal vertices $C$ be as 
described
above. 
For each equivalence class $\kappa$ represented in $C$, 
fix some vertex $x_\kappa\in M_r^G(u,X)$
for $u\in \kappa$ which is of minimum distance to $u$ 
among all such choices (for our purpose we may assume
that the empty class with 
$M_r^G(u,X)=\emptyset$ is not realized in $G$). 

\begin{figure}
\begin{center}
  \begin{tikzpicture}[circle dotted/.style={dash pattern=on .05mm 
  off 2pt, line cap=round}]
  
  \node at (-0.2, -3.3) {\textbf{a)}};
  \fill[black] (0,0) circle (2pt); 
\fill[black] (1,0) circle (2pt); 
\fill[black] (2,0) circle (2pt); 
\fill[black] (3,0) circle (2pt); 
\draw[rounded corners=10] (-0.5,-0.5) rectangle (3.5,0.5);  
\node at (0,0.25) {$x_1$};
\node at (1,0.25) {$x_2$};
\node at (2,0.25) {$x_3$};
\node at (3,0.25) {$x_4$};

\fill[black] (0.25,-1) circle (2pt); 
\fill[black] (0.5,-2) circle (2pt); 
\fill[black] (0.85,-3) circle (2pt); 
\fill[black] (1.25,-2) circle (2pt); 
\draw[-] (0.85,-3) -- (0.5,-2) -- (0.25, -1) -- (0,0); 
\draw[-] (0.5,-2) -- (1,0);
\draw[-] (1.25,-2) -- (2,0);
\draw[-] (0.85,-3) -- (1.25,-2);

\draw[-] (1.25,-3) -- (0.85,-2) -- (0.5, -1) -- (0,0); 
\draw[-] (0.85,-2) -- (1,0);
\draw[-] (0.85,-2) -- (2,0);
\draw[-] (1.75,-2) -- (2,0);
\draw[-] (1.65,-3) -- (0.85,-2);
\draw[-] (1.65,-3) -- (1.75,-2);
\fill[black] (0.5,-1) circle (2pt); 
\fill[black] (0.85,-2) circle (2pt); 
\fill[black] (1.25,-3) circle (2pt); 
\fill[black] (1.75,-2) circle (2pt); 
\fill[black] (1.65,-3) circle (2pt); 

\node at (0.85,-3.3) {$u_1$};
\node at (1.35,-3.3) {$u_2$};
\node at (1.85,-3.3) {$u_3$};

\draw[dashed] (4,1) -- (4,-4);

\begin{scope}[xshift=5cm]
  \node at (-0.2, -3.3) {\textbf{b)}};

\draw[rounded corners=10] (-0.5,-0.5) rectangle (3.5,0.5);  
\node at (0,0.25) {$x_1$};
\node at (1,0.25) {$x_2$};
\node at (2,0.25) {$x_3$};
\node at (3,0.25) {$x_4$};

\draw[-,ultra thick,red] (0.85,-3) -- (0.5,-2);
\draw[-] (0.5,-2) -- (0.25, -1) -- (0,0); 
\draw[-,ultra thick,red] (0.5,-2) -- (1,0);
\draw[-] (1.25,-2) -- (2,0);
\draw[-] (0.85,-3) -- (1.25,-2);

\draw[-,ultra thick,red] (1.25,-3) -- (0.85,-2);
\draw[-] (0.85,-2) -- (0.5, -1) -- (0,0); 
\draw[-,ultra thick,red] (0.85,-2) -- (1,0);
\draw[-] (0.85,-2) -- (2,0);
\draw[-] (1.75,-2) -- (2,0);
\draw[-,ultra thick,red] (1.65,-3) -- (0.85,-2);
\draw[-] (1.65,-3) -- (1.75,-2);

  \fill[black] (0,0) circle (2pt); 
\fill[black] (1,0) circle (2pt); 
\fill[black] (2,0) circle (2pt); 
\fill[black] (3,0) circle (2pt); 
\fill[black] (0.25,-1) circle (2pt); 
\fill[black] (0.5,-2) circle (2pt); 
\fill[black] (0.85,-3) circle (2pt); 
\fill[black] (1.25,-2) circle (2pt); 

\fill[black] (0.5,-1) circle (2pt); 
\fill[black] (0.85,-2) circle (2pt); 
\fill[black] (1.25,-3) circle (2pt); 
\fill[black] (1.75,-2) circle (2pt); 
\fill[black] (1.65,-3) circle (2pt); 

\node at (0.85,-3.3) {$u_1$};
\node at (1.35,-3.3) {$u_2$};
\node at (1.85,-3.3) {$u_3$};
\end{scope}

\draw[dashed] (9,1) -- (9,-4);

\begin{scope}[xshift=10cm]
  \node at (-0.2, -3.3) {\textbf{c)}};

\draw[rounded corners=10] (-0.5,-0.5) rectangle (3.5,0.5);  
\node at (0,0.25) {$x_1$};
\node at (1,0.25) {$x_2$};
\node at (2,0.25) {$x_3$};
\node at (3,0.25) {$x_4$};

\draw[-] (0.85,-3) -- (0.5,-2);
\draw[-] (0.5,-2) -- (0.25, -1) -- (0,0); 
\draw[-] (0.5,-2) -- (1,0);
\draw[-] (1.25,-2) -- (2,0);
\draw[-] (0.85,-3) -- (1.25,-2);

\draw[-] (1.25,-3) -- (0.85,-2);
\draw[-] (0.85,-2) -- (0.5, -1) -- (0,0); 
\draw[-] (0.85,-2) -- (1,0);
\draw[-] (0.85,-2) -- (2,0);
\draw[-] (1.75,-2) -- (2,0);
\draw[-] (1.65,-3) -- (0.85,-2);
\draw[-] (1.65,-3) -- (1.75,-2);

  \fill[black] (0,0) circle (2pt); 
\fill[black] (1,0) circle (2pt); 
\fill[black] (2,0) circle (2pt); 
\fill[black] (3,0) circle (2pt); 
\fill[black] (0.25,-1) circle (2pt); 
\fill[black] (0.5,-2) circle (2pt); 
\fill[black] (0.85,-3) circle (2pt); 
\fill[black] (1.25,-2) circle (2pt); 

\fill[black] (0.5,-1) circle (2pt); 
\fill[black] (0.85,-2) circle (2pt); 
\fill[black] (1.25,-3) circle (2pt); 
\fill[black] (1.75,-2) circle (2pt); 
\fill[black] (1.65,-3) circle (2pt); 


\draw[line width = 1pt,circle dotted] (1.05,-3.5) -- (1.25,-4) -- (1.45,-3.5);
\draw[line width = 1pt,circle dotted] (1.05, -3.5) -- (0.85,-3);
\draw[line width = 1pt,circle dotted] (1.25, -3) -- (1.45,-3.5) -- (1.65,-3);
\fill[black!30!white, draw=black] (1.05,-3.5) circle (1.5pt); 
\fill[black!30!white, draw=black] (1.45,-3.5) circle (1.5pt); 
\fill[black!30!white, draw=black] (1.25,-4) circle (1.5pt); 
\end{scope}

  \end{tikzpicture}
\end{center}
\caption{(a) The vertices $u_1,u_2,u_3$ realize the same projection
profile $\rho_r^G[u_1,X]=(3,2,2,\infty)$. (b) We have chosen 
$x_2$ as $x_\kappa$, which results in the indicated tree $T_\kappa$. (c) A subdivided copy of $T_\kappa$ is added to $\dot{G}$.}
\label{fig:dotG} 
\end{figure}
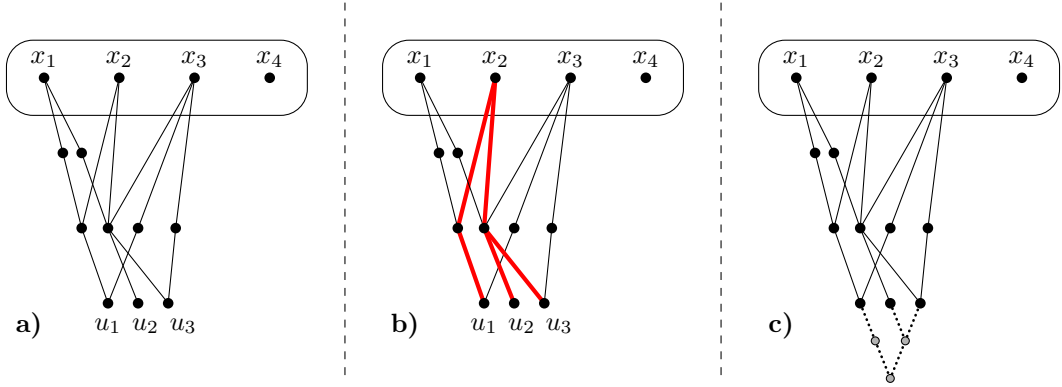

Let $T_\kappa$ be a tree which contains
for each $u\in \kappa\cap C$ an $X$-avoiding path of minimum 
length between $u$ and $x_\kappa$ (e.g.\ obtained by  
an $X$-avoiding breadth-first search with root~$x_\kappa$). 
Note that the vertices
of $\kappa\cap C$ appear as leaves of $T_\kappa$ and all 
leaves have the same distance from the root $x_\kappa$.
To see this, note that if a vertex $u$ of
$\kappa\cap C$ lies on a shortest path from~$x_\kappa$ to another
vertex $v$ of~$\kappa\cap C$, then the $X$-avoiding 
distance between $u$ and $x_\kappa$
is smaller than the $X$-avoiding distance between 
$v$ and $x_\kappa$, contradicting that all vertices of 
$\kappa\cap C$ have the same projection profile. Recall that 
by construction projection profiles are preserved for each 
vertex of $\kappa\cap C$.

Let $\dot{G}$ be the graph obtained by adding to $G'$ 
for each equivalence class $\kappa\cap C$ a 
copy of~$T_\kappa$, with each
each edge subdivided~$2r$
times. Then identify the leaves of this copy $T_\kappa$
with the respective vertices of $\kappa$.  
\end{definition}

\begin{lemma}\label{lem:classgraph}
There exists a function $f_\bullet$ such that for all 
$r'\in\N$ and all $\epsilon>0$ 
we have $\wcol_{r'}(\dot{G})\leq f_\bullet(r',t, \epsilon)
\cdot |X|^{1+\epsilon}$. 
\end{lemma}

\begin{proof}
Let $\epsilon':= \epsilon/2$. 
According to Lemma~\ref{lem:projection-complexity}, there is a
function $\fproj$ such that there are at most 
$\fproj(r,\epsilon')\cdot|X|^{1+\epsilon'}=: x$ distinct 
projection profiles. When constructing the graph~$\dot{G}$, we hence create at most so many
trees $T_\kappa$. These can be found as disjoint 
subgraphs in $G'\odot K_x$. Hence, $\dot{G}$ is a 
subgraph of a $2r$-subdivision of
$G'\odot K_x$. 
According to Lemmata~\ref{lem:wcolGX}, \ref{lem:wcollex} 
and \ref{lem:wcolsubdiv}
we have $\wcol_{r'}(\dot{G})\leq 
h(r,r',t,\epsilon')\cdot |X|^{\epsilon'}\cdot \fproj(r,\epsilon')\cdot 
|X|^{1+\epsilon'}+r'$, where $h$ is the function from 
Lemma~\ref{lem:wcolGX}. Assuming that each of these terms
is at least $1$, we can define $f_\bullet(r',t,\epsilon):=
r'\cdot h(r',t,\epsilon')\cdot \fproj(r,\epsilon')$.
\end{proof}

\begin{lemma}\label{lem:translateSteiner}
With each group Steiner tree problem for $\YYY$, we associate
the Steiner tree problem for the set~$Y$ which contains exactly the 
roots of the subdivided trees $T_\kappa$ for each 
$\kappa\in \YYY$. Denote this root by $v_\kappa$
(it is a copy of $x_\kappa$). 
Denote by $d_\kappa$ the distance from $v_\kappa$ 
to $x_\kappa$. Then every group Steiner tree $T_\YYY$ for 
$\YYY$
of size $s\leq 2t$ in $G$ gives rise to a Steiner 
tree for $Y$ of size $s+\sum_{\kappa\in \YYY} d_\kappa$ in 
$\dot{G}$. Vice versa, every Steiner tree for a set~$Y$ of
the above form
of size $s+\sum_{\kappa\in \YYY} d_\kappa$ in $\dot{G}$ 
gives rise to a group Steiner tree of size $s$ for $\YYY$
in $G$. 
\end{lemma}
\begin{proof}
The forward direction is clear. Conversely, 
let $T_Y$ be a Steiner tree for a set~$Y$ which contains
only roots of subdivided trees $T_\kappa$
of size 
$s+\sum_\kappa d_\kappa$ in $\dot{G}$. 
We claim that $T_Y$ uses exactly $d_\kappa$
vertices of $T_\kappa$, more precisely, $T_Y$ 
connects exactly one vertex $u\in \kappa$ 
with~$v_\kappa$. Assume $T_Y$ contains two paths $P_1,P_2$
between $v_\kappa$ and vertices $u_1,u_2$ from $\kappa$. 
Because we work with a $2r$-subdivision of $T_\kappa$, 
we have $|V(P_1)\cup V(P_2)|\geq d\kappa+2r$. However, 
there is a path between $u_1$ and $u_2$ via $x_\kappa$
of length at most $2r$ (which uses only $2r-1$ vertices) in 
$\dot{G}$, contradicting the fact that $T_Y$ uses
a minimum number of vertices. 
\end{proof}

\begin{lemma}
There is a function $f$ such that for every $\epsilon>0$
the graph $\dot{G}$ contains at most $f(r,t,\epsilon)\cdot 
|X|^{2+\epsilon}$ vertices. 
\end{lemma}
\begin{proof}
Let $\epsilon':=\epsilon/2$. 
Every Steiner tree $T_Y$ that connects a 
subset $Y$ decomposes into paths~$P_{uv}$ between pairs 
$u,v\in Y$. 
According to Lemma~\ref{lem:wcol-sep}, each such path $P_{uv}$ 
contains a vertex $z$ which is weakly 
$(4r^2+2t)$-reachable from $u$ and from $v$. 
This is because each Steiner tree in $\dot{G}$ connecting
$u$ and $v$ contains a path of length at most~$2r^2$ 
between $u$ and some leaf $u_\kappa\in \kappa\cap C$
(and analogously a path of length at most $2r^2$ 
between $v$ and some leaf $v_\kappa\in \kappa\cap C$). 
Now $u_\kappa$ and $v_\kappa$ are connected by a path
of length at most $2t$ by construction. 

Denote by $Q_u$
and $Q_v$, respectively, 
the sub-path of $P_{uv}$ between $u$ and $z$, and~$v$ and $z$, 
respectively. We charge the vertices of $Q_u$ to vertex $u$
and the vertices of $Q_v$ to vertex $v$ (and the vertex $z$ 
to one of the two). According to Lemma~\ref{lem:classgraph}, 
each vertex weakly $(4r^2+2t)$-reaches at most $f_\bullet(4r^2+2t,t,\epsilon')
\cdot 
|X|^{1+\epsilon'}$ vertices which can play the role of $z$. 
According to Lemma~\ref{lem:projection-complexity} we have 
at most $\fproj(r,\epsilon')\cdot |X|^{1+\epsilon'}$ choices for
$u,v\in Y$. Hence we obtain that all Steiner trees add up to at most 
$\fproj(r,\epsilon')\cdot |X|^{1+\epsilon'}\cdot f_\bullet(4r^2+2t,t,\epsilon')
\cdot 
|X|^{1+\epsilon'}=: f(r,t,\epsilon)\cdot |X|^{2+\epsilon}$
vertices. 
\end{proof}

As $G'$ is a subgraph of $\dot{G}$, we conclude
that also $G'$ is small. 

\begin{corollary}
There is a function $f$ such that for every $\epsilon>0$
the graph $\dot{G}$ has size at most $f(r,t,\epsilon)\cdot 
|X|^{2+\epsilon}$. 
\end{corollary}

This was the last missing statement of Lemma~\ref{lem:tree-closure}, 
which finishes the proof. 

\section{Lower bounds}\label{sec:lowerbounds}
Our lower bound is based on Proposition~3.2 of \cite{LokshtanovPRS16} which establishes
 equivalence between FPT-approximation
algorithms and approximate kernelization.

\begin{lemma}[Proposition~3.2 of \cite{LokshtanovPRS16}]\label{lemma:fpt-approx}
For every function $\alpha$ and decidable parameterized 
optimization problem $\Pi$,
$\Pi$ admits a fixed parameter tractable $\alpha$-approximation algorithm if and only if $\Pi$ has an $\alpha$-approximate kernel.
\end{lemma}

We will use 
a reduction from {\sc Set Cover} to the {\sc Distance-$r$ Dominating
Set} problem. Recall that the instance of the {\sc Set Cover} problem
consists of $(U, \FFF, k)$, where $U$ is a finite universe, $\FFF\subseteq 2^U$ is a family of subsets of the universe, and 
$k$ is a positive integer. The question is whether there exists a subfamily $\Gg \subseteq \FFF$ of size $k$ such that every
element of $U$ is covered by~$\Gg$, i.e., $\bigcup G=U$. 
The following result states that under complexity 
theoretic assumptions for the set cover problem
on general graphs
there does not exist a fixed-parameter tractable $\alpha$-approximation algorithm for any function $\alpha$. 

\begin{lemma}[Chalermsook et al.~\cite{chalermsook17}]\label{lemma:fpt-approx-lowerbound}
If the Gap Exponential Time Hypothesis (gap-ETH) holds, 
then there is no fixed parameter tractable $\alpha$-approximation algorithm for the
{\sc Set Cover} problem, for any function~$\alpha$.
\end{lemma}

\begin{definition}
For $p\geq 0$, let $\mathcal{H}_p$ be the class of 
exact $p$-subdivisions of all simple graphs, that is, the class 
comprising all the graphs that can be obtained from 
any simple graph by replacing every edge by a path of
length exactly $p$. 
\end{definition}


\begin{lemma}[Ne\v{s}et\v{r}il and Ossona de Mendez~\cite{nevsetvril2011nowhere}]\label{lemma:somewheredense}
For every monotone somewhere dense graph class~$\Cc$, there exists $r\in\N$ such 
that $\mathcal{H}_r\subseteq \Cc$. 
\end{lemma}

Based on the above lemma, 
in the arxiv-version of \cite{DrangeDFKLPPRVS16}, 
a parameterized reduction from {\sc Set Cover} to the
{\sc Distance-$r$ Dominating Set} problem is presented which
preserves the parameter~$k$ exactly. In that paper, the reduction
is used to prove $\mathrm{W}[2]$-hardness of 
the {\sc Distance-$r$ Dominating Set} problem. 

\begin{lemma}[Drange et al.~\cite{DrangeDFKLPPRVS16}]\label{lemma:reduction}
Let $(U,\FFF,k)$ be an instance of set cover and let 
$r\in \N$. There exists a graph $G\in \mathcal{H}_{r}$
such that $(U,\FFF,k)$ is a positive instance of the
{\sc Set Cover} problem if and only
if $(G,k)$ is a positive instance of the {\sc Distance-$r$ Dominating 
Set} problem.
\end{lemma}

Combining Lemmata \ref{lemma:fpt-approx}, \ref{lemma:fpt-approx-lowerbound}, \ref{lemma:somewheredense} and \ref{lemma:reduction} now gives the following theorem. 

\begin{theorem}
If the Gap Exponential Time Hypothesis holds, then for every
monotone somewhere dense class of graphs $\Cc$ there is no $\alpha(k)$-approximate kernel for 
the {\sc Distance-$r$ Dominating Set} problem on $\Cc$ for
any function $\alpha\colon\N\rightarrow\N$. 
\end{theorem}

The same statement holds for the {\sc Connected Distance-$r$ Dominating Set} 
problem, as every connected graph that admits a distance-$r$ dominating
set of size $k$ also admits a connected distance-$r$ dominating
set of size at most $(2r+1)k$. 

\section{Conclusion}
The study of computationally hard problems on restricted classes
of inputs is a very fruitful line of research in algorithmic graph structure
theory and in particular in parameterized complexity theory. This
research is based on the observation that many problems such as
\textsc{Dominating Set}, which are considered intractable in general,
can be solved efficiently on restricted graph classes. Of course it 
is a very desirable goal in this line of research to identify the most
general classes of graphs on which certain problems 
can be solved efficiently. In this work we were able to provide
lossy kernels for graphs of bounded expansion whose size 
matches the size of the best known kernel for {\sc Dominating Set}. 
We were furthermore able to identify
the exact limit for the existence of lossy kernels for the {\sc Connected Distance-$r$ Dominating Set} problem. 
One interesting open question
is whether our polynomial bounds on the size of the 
lossy kernel on nowhere dense classes can be improved to pseudo-linear bounds.
For $K_{d,d}$-free graphs we have an additional $\frac{1}{\alpha-1}$ multiplicative factor 
in the exponent. This leads to the question whether it is possible to reduce the size of our kernel on $K_{d,d}$-free graphs to $f(\alpha) k^{\Oh(d^2)}$ for some function $f$. 
And, in light of the $\Oh(k^{(d-1)(d-3)-\epsilon})$ lower bound for {\sc Dominating Set}, 
is it possible to obtain a lossy kernel for {\sc Dominating Set} on biclique-free graphs that beats this bound? 
Our hope is that such a ``fine-grained'' analysis of the kernelization complexity of domination problems will lead to 
a better understanding of the boundary between ``hard'' and ``easy'' instances. 

\bibliographystyle{siamplain}
\bibliography{cds_bib}

\end{document}